\documentclass[nofootinbib,twocolumn,pra,letterpaper,longbibliography]{revtex4-2}
\usepackage{latexsym,amsmath,amssymb,amsfonts,graphicx,color,amsthm}
\usepackage{enumerate}
\usepackage{enumitem}
\usepackage{braket}
\usepackage{adjustbox}
\usepackage{footnote}
\usepackage[dvipsnames]{xcolor}
\usepackage{tipa}
\usepackage{textcomp}
\usepackage{fontenc}
\usepackage{mathrsfs}
\usepackage{color}
\usepackage{colortbl}
\usepackage{graphics,graphicx}
\usepackage{txfonts}
\usepackage{lipsum, babel}
\usepackage{bbm}

\usepackage[unicode=true,pdfusetitle, bookmarks=true,bookmarksnumbered=false,bookmarksopen=false, breaklinks=false,pdfborder={0 0 0},backref=false,colorlinks=false] {hyperref}
\hypersetup{ colorlinks,linkcolor=myurlcolor,citecolor=myurlcolor,urlcolor=myurlcolor}

\definecolor{myurlcolor}{rgb}{0,0,0.7}

\usepackage{float}
\usepackage{hyperref}

\newcommand{\cH}{\mathcal{H}}

\newcommand{\cK}{\mathcal{K}}
\newcommand{\cL}{\mathcal{L}}
\newcommand{\cM}{\mathcal{M}}
\newcommand{\cN}{\mathcal{N}}

\newcommand{\cS}{\mathcal{S}}

\newcommand{\cV}{\mathcal{V}}

\newcommand{\tr}{\text{Tr}}

\newtheorem{theorem}{Theorem}
\newtheorem{proposition}{Proposition}
\newtheorem{lemma}{Lemma}

\newtheorem{definition}{Definition}
\newtheorem{example}{Example}

\newtheorem{remark}{Remark}
\newtheorem{obs}{Observation}

\begin{document}

\title{On the relevance of Hermitian-preserving trace-preserving linear maps in asymptotic information retrieval in quantum communication}

\author{Arindam Mitra$^{1,2,}$}
\email{am56@iitbbs.ac.in}
\email{arindammitra143@gmail.com}

\author{Jatin Ghai$^{3,4,}$}
\email{jghai@imsc.res.in}

\affiliation{$^1$Department of Physics, School of Basic Sciences, Indian Institute of Technology Bhubaneswar, Odisha 752050, India.\\
$^2$Korea Research Institute of Standards and Science, Daejeon 34113, South Korea.\\
$^3$Optics and Quantum Information Group, The Institute of Mathematical Sciences, C. I. T. Campus, Taramani, Chennai 600113, India.\\
$^4$Homi Bhabha National Institute, Training School Complex, Anushaktinagar, Mumbai 400094, India.
}

\date{\today}

\begin{abstract}
In quantum technologies, quantum channels are essential for transmitting quantum states. Quantum channels usually introduce noise, thereby reducing the information in a state. These are mathematically described by completely positive trace-preserving linear maps that represent the generic evolution of quantum systems and are also special cases of Hermitian-preserving trace-preserving (HPTP) linear maps. Concatenating a quantum channel with another quantum channel makes it noisier and degrades its information and resource preservability. In this work, we demonstrate a physically meaningful way to compare a pair of quantum channels using Hermitian-preserving trace-preserving linear maps. Precisely, given a pair of quantum channels $\Lambda_1, \Lambda_2$ and an arbitrary unknown input state, we show that the output state of the quantum channel $\Lambda_2$ from the pair can be uniquely identified from the output statistics of the other channel $\Lambda_1$ from the pair, using some POVM, iff $\Lambda_2$ can be obtained from $\Lambda_1$ by concatenating it with a Hermitian-preserving trace-preserving linear map that is not necessarily positive. In such cases, $\Lambda_2$ may not always be obtained from $\Lambda_1$ through post-processing. The relation between these channels is a preorder, and we study its characterization in this work. In other words, we provide an operational interpretation of this preorder induced by HPTP linear maps in a communication scenario. Furthermore, we characterize the difficulty of implementing $\Lambda_2$ given that $\Lambda_1$ has already been implemented via a quantifier, namely, physical implementability. We also illustrate the implications of our results for the incompatibility of quantum devices through an example. In short, we provide valuable insights about the relevance of Hermitian-preserving trace-preserving linear maps in physically motivated settings.
\end{abstract}

\maketitle
\section{Introduction}

In quantum theory, quantum channels, i.e., completely positive and trace-preserving (CPTP) linear maps, are fundamental objects that provide the mathematical framework for describing all physically realizable deterministic transformations of quantum systems \cite{wilde_2013_quantum,nielsen_2010_quantum}. They encompass how quantum systems evolve while interacting with the environment\cite{breuer_2002_open,lidar_2020_lec_notes}, making them a cornerstone of quantum theory for quantum communication, information processing, and computation. The structure of quantum channels provides us with tools for the characterization of noise in relevant experimental scenarios, thus bridging them with the abstract quantum mechanics. When a quantum channel is post-processed with another channel, the resultant channel is noisier. Thus, its ability to communicate information, both classical and quantum, is degraded. This is reflected by a decrease in its classical and quantum capacities\cite{wilde_2013_quantum}.

On the other hand, compared to quantum channels, Hermitian-preserving trace-preserving (HPTP) linear maps have been less explored from a physical perspective. A basic requirement for a transformation in quantum theory is that it should map real-valued observables to real-valued observables. Thus, for any linear map describing these kinds of transformations, Hermitian-preservability is the weakest natural demand. Among Hermitian-preserving trace-preserving linear maps, only CPTP linear maps represent the generic evolution of quantum systems. However, the physical implementability of Hermitian-preserving trace-preserving linear maps has also been a subject of investigation\cite{Jiang_2021_physical,Zhao_2023_information,Zhao_2025_sim,Parzygnat_2024_broad,Rossini_2023_NonMarkov}. Recently, an algorithm has also been developed for the implementation of these general Hermitian-preserving maps through exponentiation\cite{Wei_2024_Herm}. Besides this, positive maps are used for the construction of the witnesses\cite{Peres_SEP_cri,Chruściński_2014_wit,Mallick_2025_ent,Huber_2014_ent,Rodríguez_2008_Pos,Li_2024_ent_det} for entanglement detection.  Thus, there is quite a bit of interest in exploring the transformations beyond complete positivity, as this condition is deemed a bit restrictive, especially in quantum open system scenarios\cite{Carteret_beyond_CP,Dominy_Beyond_CP,Pastuszak_herm_pres,Rossini_2023_NonMarkov,Usha_2011_open}. In this context, the existence of quantum maps between CPTP linear maps and Hermitian-preserving trace-preserving linear maps has also been studied\cite{cao_2023_hptp}.  \emph{Motivated by these}, in this work, we provide \emph{a physically meaningful way} to compare two quantum channels using Hermitian-preserving trace-preserving linear maps. More specifically, given a pair of quantum channels $\Lambda_1$ and $\Lambda_2$, we show that the measurement statistics obtained from the output of $\Lambda_1$ can uniquely determine the output state of $\Lambda_2$ for every input state, if and only if $\Lambda_2$ can be expressed as the composition of $\Lambda_1$ with a Hermitian-preserving, trace-preserving linear map that is not necessarily positive. Thus, we establish the operational significance of HPTP linear maps in quantum communication scenarios. Through a simple example, we illustrate that in such a scenario, $\Lambda_2$ can't always be obtained from $\Lambda_1$ through post-processing it with another quantum channel. Both of these comparative relations are distinct and follow a hierarchy. We then characterize the difficulty of implementing the former channel given that the latter channel has already been implemented via a quantifier, namely, \emph{physical implementability}.  We also highlight the implications of our results on the incompatibility of quantum devices.

The rest of this paper is arranged as follows. In Sec. \hyperref[Sec:Prelim]{II}, we discuss the preliminaries. 
In Sec. \hyperref[Sec:Main_res]{III}, we present the main results of this work. More specifically, in Sec. \hyperref[Sec:Main_res_asymp]{III.A}, we provide \emph{a detailed analysis} of the comparison of quantum channels using Hermitian-preserving trace-preserving linear maps in a physically meaningful way. For a given pair of channels, we provide an \emph{operational relevance} of HPTP maps in terms of asymptotic information retrieval of the output of one channel from the output of the other. In Sec. \hyperref[Sec:Main_res_Hierarchy]{III.B} we discuss the hierarchy between the set of channels obtained from a given channel by post-processing it with Hermitian-preserving trace-preserving linear maps and the set of channels obtained by post-processing the same given channel with other quantum channels. The set of channels obtained by post-processing the same given channel with positive maps lies between these two sets. In Sec. \hyperref[Sec:Main_res_Phys_Imp]{III.C}, we discuss the difficulty of \emph{physical implementability} of $\Lambda_2$ in scenarios where $\Lambda_1$ has already been implemented and $\Lambda_2$ can be written as the concatenation of $\Lambda_1$ with an HPTP linear map. In Sec. \hyperref[Sec:Conc]{IV}, we summarize our results and discuss future directions.

\section{Preliminaries}\label{Sec:Prelim}
\subsection{Quantum Measurements, Channels and Instruments}
A quantum measurement or POVM constitutes a set of positive-semidefinite matrices $M=\{M(x)\in\cL(\cH)\}_{x\in\Omega_M}$ acting on a Hilbert space $\cH$ satisfying the condition $\sum_{x\in\Omega_M}M(x)=\mathbbm{1}_{\cH}$ where $\mathbbm{1}_{\cH}$ is the identity operator on the Hilbert space $\cH$\cite{Heinosaari_book_QF}. Here, $\Omega_M$ is the set of outcomes for the POVM $M$ and $\cL(\cH)$ is the set of all linear operators on the Hilbert space $\cH$. In this work, we assume the dimension of $\cH$ to be finite.  Each $M(x)$ is called a POVM element of the POVM $M$. If we have $M^2(x)=M(x)~\forall x\in\Omega_M$, then $M$ is said to be a projective measurement. In this work, we will only consider POVMs that have a finite number of outcomes. If POVM $M$ is performed on a quantum state $\rho\in\cS(\cH)$, where $\cS(\cH)$ is the set of all density matrices on $\cH$, then the probability of obtaining the outcome $x$ is given by $\tr[\rho M(x)]$. The set of all the POVMs acting on the Hilbert space $\cH$ is denoted as $\mathscr{M}(\cH)$.

A quantum channel $\Lambda:\cL(\cH)\rightarrow\cL(\cK)$ is mathematically represented by a completely positive and trace-preserving (CPTP) linear map\cite{Heinosaari_book_QF}. Physically, it transforms an arbitrary density matrix to another arbitrary density matrix. This description is known as the Schrödinger picture. For any linear map $\Phi:\cL(\cH)\rightarrow\cL(\cK)$, the dual map $\Phi^{\dagger}:\cL(\cK)\rightarrow\cL(\cH)$ is defined through the equation
\begin{equation}
    \tr[B^{\dagger}\Phi^{\dagger}(A)]=\tr[(\Phi(B))^{\dagger}A],
\end{equation}
where $A\in\cL(\cK)$ and $B\in\cL(\cH)$\cite{watrous_2018_theory}. The dual map of a quantum channel represents the action of that channel in the Heisenberg picture. We denote the set of all quantum channels with the input space $\cL(\cH)$ and the output space $\cL(\cK)$ as $\mathscr{C}(\cH,\cK)$. The composition of two arbitrary linear maps $\Lambda_1:\cL(\cH_2)\rightarrow\cL(\cH_1)$ and $\Lambda_2:\cL(\cH)\rightarrow\cL(\cH_2)$ is denoted as $\Lambda_1\circ\Lambda_2$ where for an arbitrary $\rho\in\cL(\cH)$, we have $\Lambda_1\circ\Lambda_2(\rho):=\Lambda_1(\Lambda_2(\rho))$. Clearly, $\Lambda_1\circ\Lambda_2\in\mathscr{C}(\cH,\cH_1)$. For two arbitrary quantum channels $\Lambda\in\mathscr{C}(\cH,\cK)$ and $\Lambda^{\prime}\in\mathscr{C}(\cH,\cK^{\prime})$ if $\Lambda=\Theta\circ\Lambda^{\prime}$ where $\Theta\in\mathscr{C}(\cK^{\prime},\cK)$ then we denote it as $\Lambda^{\prime}\succeq_{postproc}\Lambda$. Here, the relation $\succeq_{postproc}$ is a preorder, which is also known as the \emph{post-processing preorder} \cite{Heinosaari_incomp_chan}. If $\Lambda\succeq_{postproc}\Lambda^{\prime}$ also holds, then $\Lambda$ and $\Lambda^{\prime}$ are called \emph{post-processing equivalent}. They are represented as $\Lambda\simeq_{postproc}\Lambda^{\prime}$. The action of a quantum channel $\Lambda\in\mathscr{C}(\cH,\cK)$ on a POVM $M=\{M(x)\in\cL(\cK)\}$ in the Heisenberg picture is denoted as $\Lambda^{\dagger}(M):=\{\Lambda^{\dagger}(M(x))\}$.


A linear map $\Lambda:\cL(\cH)\rightarrow\cL(\cK)$ is said to be Hermitian-preserving if for arbitrary $X\in Herm(\cH)$ we have $\Lambda(X)\in Herm(\cK)$\cite{watrous_2018_theory}. Here, $Herm(\cH)$ is the set of Hermitian matrices in $\cL(\cH)$. Let us denote the set of all \emph{Hermitian-preserving trace-preserving linear maps} with the input space $\cL(\cH)$ and the output space $\cL(\cK)$ as $\mathscr{C}_{HP}(\cH,\cK)$. Similar to the post-processing preorder, if a quantum channel $\Lambda_1\in\mathscr{C}(\cH,\cK_1)$ is related to another quantum channel $\Lambda_2\in\mathscr{C}(\cH,\cK_2)$ through the relation $\Lambda_2=\Theta\circ\Lambda_1$, where $\Theta\in\mathscr{C}_{HP}(\cK_1,\cK_2)$, it is denoted as $\Lambda_1\succeq_{HP}\Lambda_2$. It can be easily seen that the relation $\succeq_{HP}$ is also a preorder.

A linear map $\Lambda:\cL(\cH)\rightarrow\cL(\cK)$ is called positive if for any $X\in \cL_{+}(\cH)$ we have $\Lambda(X)\in\cL_{+}(\cK)$, where $\cL_{+}(\cH)$  is the set of positive-semidefinite matrices in $\cL(\cH)$\cite{watrous_2018_theory}. $\mathscr{C}_P(\cH,\cK)$ denotes the set of all \emph{positive trace-preserving linear maps} with the input space $\cL(\cH)$ and the output space $\cL(\cK)$. Similarly as above, if a quantum channel $\Lambda_1\in\mathscr{C}(\cH,\cK_1)$ is related to another quantum channel $\Lambda_2\in\mathscr{C}(\cH,\cK_2)$ through the relation $\Lambda_2=\Theta\circ\Lambda_1$, where $\Theta\in\mathscr{C}_{P}(\cK_1,\cK_2)$, it is denoted as $\Lambda_1\succeq_{P}\Lambda_2$. It can also be easily seen that the relation $\succeq_{P}$ is a preorder.  It is worth mentioning that $\mathscr{C}(\cH,\cK)\subset\mathscr{C}_{P}(\cH,\cK)\subset\mathscr{C}_{HP}(\cH,\cK)$\cite{watrous_2018_theory,nielsen_2010_quantum}. 

For a linear map $\Lambda:\cL(\cH)\rightarrow\cL(\cK)$ there exists an equivalent representation known as Choi-Jamiolkowski isomorphism\cite{watrous_2018_theory}. Corresponding to this linear map, an operator $J_{\Lambda}\in\cL(\cK\otimes\cH)$, known as the Choi matrix of the linear map $\Lambda$, is defined as
\begin{align}
    J_{\Lambda}:=(\Lambda\otimes\mathbbm{I}_\cH)\ket{\Omega}\bra{\Omega}.
\end{align}
Here, $\mathbbm{I}_\cH\in\mathscr{C}(\cH,\cH)$ is the identity channel along with $\ket{\Omega}=\sum_{i=1}^{d_\cH}\ket{i}\otimes\ket{i}$, where $d_\cH$ is the dimension of the Hilbert space $\cH$ and $\{\ket{i}\}_{i=1}^{d_\cH}$ is an orthonormal basis in the Hilbert space $\cH$. Specific properties of a linear map are reflected in the properties of its Choi matrix. If $\Lambda$ is a trace-preserving linear map then $\tr_{\cK}[J_\Lambda]=\mathbbm{1}_{\cH}$. Additionally, $\Lambda$ is Hermitian-preserving if and only if $J_\Lambda^\dagger=J_\Lambda$ and $\Lambda$ is completely positive if and only if $J_\Lambda\geq0$. For composition of two arbitrary linear maps $\Lambda_1:\cL(\cH_2)\rightarrow\cL(\cH_1)$ and $\Lambda_2:\cL(\cH)\rightarrow\cL(\cH_2)$ denoted as $\Lambda_1\circ\Lambda_2$, the Choi matrix is given as\cite{Chiribella_quant_combs} $J_{\Lambda_1\circ\Lambda_2}=\tr_{\cH_2}[(J_{\Lambda_1}\otimes\mathbbm{1}_{\cH})(\mathbbm{1}_{\cH_1}\otimes J_{\Lambda_2}^{T_{\cH_2}})]$. Here, $T_{\cH_2}$ is the partial transpose with respect to the subsystem $\cH_2$.

Both quantum measurements and quantum channels can be simultaneously generalized to the notion of quantum instruments. Mathematically, a quantum instrument $\mathbf{I}$ is defined as a set of CP trace non-increasing linear maps, i.e., $\mathbf{I}=\{\Phi_x:\cL(\cH)\rightarrow\cL(\cK)\}_{x\in\Omega_\mathbf{I}}$ such that $\Phi=\sum_{x\in\Omega_\mathbf{I}}\Phi_x$ is a quantum channel \cite{Heinosaari_book_QF}. For the instrument $\mathbf{I}$ acting on a quantum state $\rho$, the classical output of the instrument is denoted by $x$, while its quantum output is given by $\frac{\Phi_x(\rho)}{\tr[\Phi_x(\rho)]}$. Both of these occur with probability $\tr[\Phi_x(\rho)]$. The set of all such quantum instruments with the input Hilbert space $\cH$ and the output Hilbert space $\cK$ is denoted as $\mathscr{I}(\cH,\cK)$. A quantum instrument induces a unique POVM $M=\{M(x)\}$ such that $\tr[\Phi_x(\rho)]=\tr[\rho M(x)]$ for all $\rho\in\cL(\cH)$ and $x\in\Omega_{\mathbf{I}}$ or equivalently, $M(x)=\Phi_x^{\dagger}(\mathbbm{1}_{\cK}) ~\forall~x\in\Omega_\mathbf{I}$ through duality. For a given POVM, there exist many different instruments that implement it.

In the most basic sense of the word, the compatibility of a set of quantum devices means their ability to be performed simultaneously. In other words, there exists a joint device that reproduces the outcomes of all devices in that set. A set of devices is said to be compatible if such a joint device exists for them. Otherwise, the set is incompatible. In the context of quantum measurements and channels, the following three notions of compatibility have already been defined \cite{Heinosaari_incomp_review}:
\begin{itemize}
    \item \textit{Measurement compatibility:} Two arbitrary POVMs $M\in\mathscr{M}(\cH)$ and $N\in\mathscr{M}(\cH)$ are said to be compatible if there exists a joint POVM $G=\{G(x,y)\in\cL(\cH)\}\in\mathscr{M}(\cH)$ with outcome set $\Omega_M\times\Omega_N$ such that \cite{Heinosaari_incomp_review}
\begin{align}
    M(x)&=\sum_{y\in\Omega_N}G(x,y)\qquad\forall x\in\Omega_M,\\
    N(y)&=\sum_{x\in\Omega_M}G(x,y)\qquad\forall y\in\Omega_N.
\end{align}
By performing the POVM $G$, the outcomes of both $M$ and $N$ can be obtained simultaneously. This notion is naturally generalized to an arbitrary number of POVMs.
\item \textit{Channel compatibility:} Two channels $\Lambda_1\in\mathscr{C}(\cH,\cK_1)$ and $\Lambda_2\in\mathscr{C}(\cH,\cK_2)$ are said to be compatible if there exists a joint channel $\Lambda\in\mathscr{C}(\cH,\cK_1\otimes\cK_2)$ such that \cite{Heinosaari_incomp_chan,Heinosaari_incomp_review} 
\begin{align}
    \Lambda_1&=\tr_{\cK_2}\circ\Lambda,\nonumber\\
    \Lambda_2&=\tr_{\cK_1}\circ\Lambda.
\end{align}
This definition of compatibility can also be naturally extended to an arbitrary set of channels.
\item \textit{Measurement-Channel compatibility:} A POVM $M=\{M(x)\}\in\mathscr{M}(\cH)$ and a channel $\Lambda\in\mathscr{C}(\cH,\cK)$ are said to be compatible if there exists a quantum instrument\cite{Heinosaari_book_QF} $\mathbf{I}=\{\Phi_x\}\in\mathscr{I}(\cH,\cK)$ such that for any $\rho\in\cS(\cH)$ we have $\tr[\Phi_x(\rho)]=\tr[M(x)\rho]~ \forall x$ and $\sum_x\Phi_x=\Lambda$. Thus, by implementing the instrument $\mathbf{I}$, the POVM $M$ and the channel $\Lambda$ can be implemented simultaneously.
\end{itemize}

It is known that for a given quantum channel $\Lambda\in\mathscr{C}(\cH,\cK)$ there exists its Stinespring dilation\cite{stinespring_dil,watrous_2018_theory,nielsen_2010_quantum} $(\cV,\overline{\cK})$, where $\cV:\cH\rightarrow\cK\otimes\overline{\cK}$ is an isometry, such that
\begin{align}
    \Lambda=\tr_{\overline{\cK}}[\cV\rho\cV^{\dagger}],
\end{align}
for every $\rho\in\cL(\cH)$. Through this dilation we can define the \emph{conjugate channel} $\overline{\Lambda}\in\mathscr{C}(\cH,\overline{\cK})$ as
\begin{align}
    \overline{\Lambda}=\tr_{\cK}[\cV\rho\cV^{\dagger}].
\end{align}
Depending on the isometry, there exist multiple conjugate channels for a given quantum channel, and they are post-processing equivalent\cite{Heinosaari_2017_INC}. We state the following two propositions from Refs. \cite{Heinosaari_2017_INC,Heinosaari_2018_INC}  without proof, which will be useful later:
\begin{proposition}
Consider two quantum channels $\Lambda_1\in\mathscr{C}(\cH,\cK_1)$ and $\Lambda_2\in\mathscr{C}(\cH,\cK_2)$. Then the following three statements are equivalent\cite{Heinosaari_2017_INC}:
\begin{enumerate}
    \item $\Lambda_1$ and $\Lambda_2$ are compatible.\label{Prop:chan_chan_incomp_point_1}
    \item $\overline{\Lambda}_2\succeq_{postproc}\Lambda_1$.\label{Prop:chan_chan_incomp_point_2}
    \item $\overline{\Lambda}_1\succeq_{postproc}\Lambda_2$.\label{Prop:chan_chan_incomp_point_3}
\end{enumerate}
\label{Prop:chan_chan_incomp}
\end{proposition}

\begin{proposition}
    Consider a channel $\Lambda\in\mathscr{C}(\cH,\cK)$ and a POVM $M\in\mathscr{M}(\cH)$. They are compatible with each other if and only if the relation $M(x)=\overline{\Lambda}^{\dagger}(M'(x))$ holds. Here, $\overline{\Lambda}^{\dagger}:\cL(\overline{\cK})\rightarrow\cL(\cH)$ is the dual of the conjugate channel $\overline{\Lambda}\in\mathscr{C}(\cH,\overline{\cK})$ and $M'=\{M'(x)\}\in\mathscr{M}(\overline{\cK})$ is a valid POVM on $\overline{\cK}$\cite{Heinosaari_2018_INC}.\label{Prop:meas_chan_incom}
\end{proposition}
Although while stating Proposition \ref{Prop:meas_chan_incom}, the authors of Ref. \cite{Heinosaari_2018_INC} referred to the minimal conjugate channel, in the same work, they have also stated that the minimality condition is not necessary.
\subsection{Informationally complete POVMs}
An important class of POVMs is the informationally complete POVMs\cite{Ariano_IC_2004}. It is a special kind of POVM whose output statistics uniquely determine the state of the quantum system. If $d$ is the dimension of the system Hilbert space $\cH$, then at least $d^2$ POVM elements are required to form an informationally complete POVM. Any linear operator $A\in\cL(\cH)$ can be expanded in terms of informationally complete POVM $M=\{M(x)\}_{x=1}^N$ as 
\begin{equation}
    A=\sum_{x=1}^N c_x M(x),\label{Eq:IC_POVM_Exp}
\end{equation}
where $N\geq d^2$. In other words, the POVM elements $\{M(x)\}$ of an informationally complete POVM span $\cL(\cH)$. It is worth mentioning that even if the operator $A$ is positive-semidefinite, some of the expansion coefficients $\{c_x\}$ in Eq. \eqref{Eq:IC_POVM_Exp} can be negative. The POVM elements of any informationally complete POVM form an overcomplete spanning set in general. Thus, the expansion in Eq. \eqref{Eq:IC_POVM_Exp} is not unique in general. An informationally complete POVM is said to be minimal if $N=d^2$, \textit{i.e.} there are exactly $d^2$ linearly independent POVM elements spanning the operator space $\cL(\cH)$. Thus, the expansion in Eq. \eqref{Eq:IC_POVM_Exp} is unique for the POVM elements of a minimal informationally complete POVM.

\subsubsection{State reconstruction from informationally complete POVM}
Any density matrix $\rho\in\cS(\cH)$ can be expanded in terms of a minimal informationally complete POVM $M=\{M(x)\}$ as
\begin{align}
    \rho=\sum_y^{d^2}c_yM(y).\label{Eq:den_exp_MIC}
\end{align} 
Further, we can write
\begin{align}
    p(x):=&\tr[\rho M(x)]=\sum_y^{d^2}c_y \tr[M(x)M(y)]\label{Eq:prob_IC_meas}\\
    =&\sum_yR(x,y)c_y,
\end{align}
where $R(x,y):=\tr[M(x)M(y)]$. We can rewrite the above expression in the form of a matrix equation as
\begin{align}
    P=RC,
\end{align}
where $P$ is a $d^2\times1$ matrix having entries $p(x)$, $R$ is a $d^2\times d^2$ matrix having entries $R(x,y)$, and $C$ is again a $d^2\times1$ matrix having entries $c_y$. Now, let us assume there exists a nonzero vector represented by a $d^2\times1$ matrix $\mathbf{v}$ having entries $v_y$ corresponding to a zero eigenvalue of the matrix $R$. Then, we have
\begin{align}
    &R\mathbf{v}=0\nonumber\\
   \Rightarrow &\sum_y R(x,y)v_y=0~\forall~x\nonumber\\
   \Rightarrow &\tr[M(x)\sum_y v_yM(y)]=0~\forall~x.\label{Eq:MIC_span}
\end{align}
As we know that the set $\{M(x)\}$ spans $\cL(\cH)$, from Eq. \eqref{Eq:MIC_span} we can conclude that
\begin{align}
        \sum_y v_yM(y)=&0.
\end{align}
Since the informationally complete POVM $M$ is minimal, the set $\{M(y)\}$ is linearly independent. Hence
\begin{align}
    \sum_y v_y M(y)=0\quad\Rightarrow\quad v_y=0,~ \forall~ y,
\end{align}
which is a contradiction. Thus, the matrix $R$ is full rank, and hence, it is invertible. Therefore, we can always write
\begin{align}
    C=R^{-1}P\label{Eq:state_recon_Ic}.
\end{align}

Thus, by obtaining the probability distribution given in Eq. \eqref{Eq:prob_IC_meas} through performing the minimal informationally complete POVM $M$ on \emph{an arbitrary unknown} quantum state $\rho$, one can \emph{always reconstruct} the state $\rho$ using Eq. \eqref{Eq:den_exp_MIC} and  Eq. \eqref{Eq:state_recon_Ic}. Similar logic follows for reconstructing an arbitrary unknown quantum state by performing a non-minimal informationally complete POVM. 

\section{Asymptotic information retrieval in quantum communication}\label{Sec:Main_res}

Consider a communication scenario depicted in Fig. \ref{Fig:In_retrive_game} where Alice prepares an arbitrary quantum state $\rho$ and sends it to Bob through the quantum channel $\Lambda_1$. She can send as many copies as he demands. Bob can perform an arbitrary POVM of the state $\Lambda_1(\rho)$ that he receives from Alice and therefore is able to perform full state tomography. Furthermore, $\rho$ is unknown to Bob. Now, the question is: When can Bob perfectly identify $\Lambda_2(\rho)$ in the asymptotic sense, i.e., given a large enough number of copies of $\Lambda_1(\rho)$? Note that Bob is not required to identify $\rho$. Instead, it is enough for him to identify $\Lambda_2(\rho)$. A trivial scenario is when there exists a channel $\widetilde{\Lambda}$ such that $\Lambda_2=\widetilde{\Lambda}\circ\Lambda_1$. Then Bob can implement the channel $\widetilde{\Lambda}$ on the state received from Alice and perform full state tomography to identify the quantum state. But is it possible to perfectly identify $\Lambda_2(\rho)$ if such a $\widetilde{\Lambda}$ does not exist? In this section, we study this and show how this asymptotic information retrieval provides an operational notion of a certain preorder of quantum channels.

\subsection{Comparison of quantum channels through the communication task given in Fig. \ref{Fig:In_retrive_game}}\label{Sec:Main_res_asymp}

\begin{figure}[!h]
    \centering
    \includegraphics[width=\linewidth]{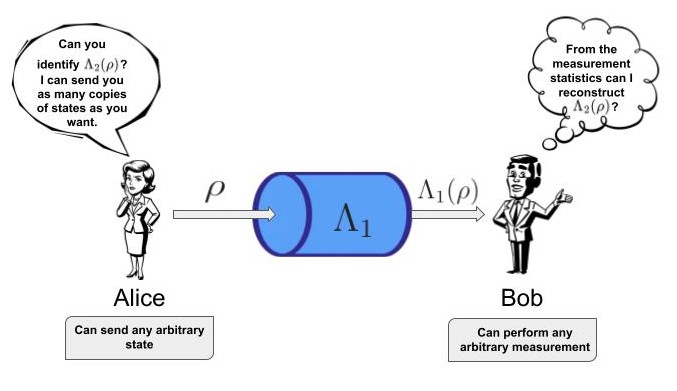}
    \caption{This figure illustrates the communication scenario discussed in Sec. \ref{Sec:Main_res_asymp}. Alice prepares a quantum state $\rho$ and sends it to Bob through the quantum channel $\Lambda_1$. She then asks Bob to predict the output $\Lambda_2(\rho)$, where the input $\rho$ is unknown to Bob. Alice agrees to send as many copies of the state $\rho$ through the channel $\Lambda_1$ as Bob demands, and Bob has the ability to perform arbitrary POVMs on the state that he receives from Alice. Now, the question is: Under what conditions can Bob perfectly identify $\Lambda_2(\rho)$ in the asymptotic sense, i.e., given an arbitrarily large number of copies of $\Lambda_1(\rho)$? It is important to emphasize that Bob's task is not to identify $\rho$. Rather, his goal is to determine the output state $\Lambda_2(\rho)$ solely from the information contained in the output of $\Lambda_1$.}
    \label{Fig:In_retrive_game}
\end{figure}

\begin{definition}
    A channel $\Lambda_1:\cL(\cH)\rightarrow\cL(\cK_1)$ is said to be at least as powerful as another quantum channel $\Lambda_2:\cL(\cH)\rightarrow\cL(\cK_2)$ in the asymptotic sense if there exists an informationally complete POVM $M_2\in\mathscr{M}(\cK_2)$ and another POVM $M_1\in\mathscr{M}(\cK_1)$ with $\Omega_{M_2}=\Omega_{M_1}$ such that for all $\rho\in\cL(\cH)$ we have 
    \begin{equation}
        \tr[\Lambda_2(\rho)M_2(x)]=\tr[\Lambda_1(\rho)M_1(x)]~\forall x\in\Omega_{M_2} \label{Eq:def_chan_power}.
    \end{equation}
    We denote it as $\Lambda_1\succeq_{asymp}\Lambda_2$.
    \label{Def:chan_power_IC}
\end{definition}

If both $\Lambda_1\succeq_{asymp}\Lambda_2$ and $\Lambda_2\succeq_{asymp}\Lambda_1$ hold, we denote it as $\Lambda_2\simeq_{asymp}\Lambda_1$. In the Heisenberg picture, Eq. \eqref{Eq:def_chan_power} can be written as

\begin{equation}
    \Lambda_2^{\dagger}(M_2(x))=\Lambda_1^{\dagger}(M_1(x))~\forall x\in\Omega_{M_2}.\label{Eq:def_chan_power_HP}
\end{equation}

\begin{remark}[Justification of Definition \ref{Def:chan_power_IC}]
    \rm{For an arbitrary unknown input state $\rho\in\mathcal{S}(\mathcal{H})$, consider the following operational procedure. First, the channel $\Lambda_1$ is applied to the quantgum state $\rho$, and the POVM $M_1$ is performed on the resulting state to obtain the corresponding probability distribution. By estimating this distribution from a sufficiently large number of independent copies of $\Lambda_1(\rho)$, one can reconstruct the output state $\Lambda_2(\rho)$ using Eqs. \eqref{Eq:state_recon_Ic}, provided that the condition in Eq. \eqref{Eq:def_chan_power} is satisfied. Thus, for any unknown input state $\rho$, the output of $\Lambda_2$ can be reconstructed using only the measurement statistics obtained from the output of $\Lambda_1$. Therefore, in this operational setting, $\Lambda_1$ is at least as powerful as $\Lambda_2$ in the asymptotic sense, which justifies Definition \ref{Def:chan_power_IC}. Note, that if $\Lambda_1\succeq_{asymp}\Lambda_2$ then the communication task given in Fig. \ref{Fig:In_retrive_game} can be performed perfectly by Bob.}\label{Rem:just_asymp}
\end{remark}

Note that in Definition \ref{Def:chan_power_IC}, the POVM $M_2$ may not be minimal informationally complete. Therefore, we state and prove the following proposition.

\begin{proposition}
        A channel $\Lambda_1:\cL(\cH)\rightarrow\cL(\cK_1)$ is at least as powerful as another quantum channel $\Lambda_2:\cL(\cH)\rightarrow\cL(\cK_2)$ in the asymptotic sense (i.e., $\Lambda_1\succeq_{asymp}\Lambda_2$) if and only if there exists a minimal informationally complete POVM $\overline{M}_2\in\mathscr{M}(\cK_2)$ and another POVM $\overline{M}_1\in\mathscr{M}(\cK_1)$ with $\Omega_{\overline{M}_2}=\Omega_{\overline{M}_1}$ such that for all $\rho\in\cL(\cH)$ we have 
    
    \begin{equation}
        \tr[\Lambda_2(\rho)\overline{M}_2(x)]=\tr[\Lambda_1(\rho)\overline{M}_1(x)]~\forall x\in\Omega_{\overline{M}_2} \label{Eq:prop_min_IC_chan_power}.
    \end{equation}
    \label{Prop:chan_power_min_IC}
\end{proposition}

\begin{proof}
    If the channel $\Lambda_1$ is at least as powerful as the channel $\Lambda_2$ in the asymptotic sense, then from Definition \ref{Def:chan_power_IC}, we have an informationally complete POVM $M_2\in\mathscr{M}(\cK_2)$ and another POVM $M_1\in\mathscr{M}(\cK_1)$ such that Eq. \eqref{Eq:def_chan_power} holds. As $M_2$ is informationally complete the $|\Omega_{M_2}|\geq d^2$ where $d=\dim(\cK_2)$ and $|\Omega_{M_2}|$ is the cardinality of the set $\Omega_{M_2}$. If $M_2$ is minimal informationally complete, then the Proposition trivially holds. If $M_2$ is not minimal informationally complete then $|\Omega_{M_2}|> d^2$. Let $|\Omega_{M_2}|=d^2+n$ where $n>0$ is a natural number. Among the POVM elements of $M_2$, $d^2$ elements are linearly independent and are able to span $\cL(\cK_2)$. Without loss of generality, we can assume that these linearly independent POVM elements are $M_2(1),\ldots M_2(d^2)$. Then the matrix $X:=\sum^{d^2+n}_{i=d^2+1}M_2(i)$ can be written as $X=\sum^{d^2}_{i=1}c_iM_2(i)$. Note that $X\geq 0$ and therefore, although some of the $c_i$s can be negative, at least one of those $c_i$s is strictly positive. Let $c_j$ be one of such strictly positive $c_i$s, i.e., $c_j>0$. In case of more than one positive $c_i$s, $c_j$ can be arbitrarily chosen among those positive $c_i$s. Let $\overline{M}_2=\{\overline{M}_2(i)\}$ be a set of $d^2$ matrices such that $\overline{M}_2(i)=M_2(i)~\forall i\in\{1,\ldots,d^2\}\setminus j$ and $\overline{M}_2(j)=M_2(j)+X$. Clearly, $\overline{M}_2(i)\geq 0~\forall i\in\{1,\ldots,d^2\}$ and $\sum^{d^2}_{i=1}\overline{M}_2(i)=\sum_iM_2(i)=\mathbbm{1}_{\cK_2}$. Hence, $\overline{M}_2\in\mathscr{M}(\cK_2)$ is a valid POVM with $|\Omega_{\overline{M}_2}|=d^2$. Now, note that the set $\{M_2(1),\ldots M_2(d^2)\}$ is linearly independent. Consider $\{\alpha_1,\ldots,\alpha_{d^2}\}$ be an arbitrary set of complex numbers i.e., $\alpha_i\in\mathbbm{C}~\forall i$. Then
    \begin{align}
        &\sum_i\alpha_i\overline{M}_2(i)=0,\nonumber\\
        \Rightarrow&\sum_{i,i\neq j}\alpha_{i}\overline{M}_2(i)+\alpha_j\overline{M}_2(j)=0,\nonumber
        \end{align}
        \begin{align}
         \Rightarrow&\sum_{i,i\neq j}\alpha_{i}M_2(i)+\alpha_j[M_2(j)+X]=0,\nonumber\\
        \Rightarrow&\sum_{i,i\neq j}\alpha_{i}M_2(i)+\alpha_j[M_2(j)+\sum^{d^2}_{i=1}c_iM_2(i)]=0,\nonumber\\
        \Rightarrow&\sum_{i}(\alpha_{i}+\alpha_jc_{i})M_2(i)=0,\nonumber\\
        \Rightarrow&(\alpha_{i}+\alpha_jc_{i})=0~\forall i.
    \end{align}
    Let us define $\beta_{i}:=\alpha_{i}+\alpha_jc_{i}$. Note that $c_j>0$ and therefore, $\beta_{j}=\alpha_j(1+c_j)=0\Rightarrow \alpha_j=0$.  As $\alpha_j=0$, $\beta_{i}=0\Rightarrow \alpha_i=0$ for all $i\in\{1,\ldots,d^2\}\setminus j$. Hence, $\alpha_i=0\forall i$. Therefore, the set $\{\overline{M}_2(1),\ldots \overline{M}_2(d^2)\}$ is linearly independent. Clearly, for an arbitrary matrix $Y\in\cL(\cK_2)$ and $Y=\sum^{d^2}_{i=1}c_i^{\prime}M_2(i)$, one can write $Y=\sum^{d^2}_{i=1}c_i^{\prime\prime}\overline{M}_2(i)$ where $c^{\prime\prime}_i=c^{\prime}_i-\frac{c_ic^{\prime}_j}{1+c_j}~\forall i\in\{1,\ldots,d^2\}\setminus j$ and $c^{\prime\prime}_j=\frac{c^{\prime}_j}{1+c_j}$, i.e., the set $\{\overline{M}_2(1),\ldots,\overline{M}_2(d^2)\}$ is able to span $\cL(\cK_2)$. Hence, $\overline{M}_2$ is a minimal informationally complete POVM. Consider another set of matrices $\overline{M}_1=\{\overline{M}_1(1),\ldots,\overline{M}_1(d^2)\}$ such that $\overline{M}_1(i)=M_1(i)~\forall i\in\{1,\ldots,d^2\}\setminus j$ and $\overline{M}_1(j)=M_1(j)+\sum^{d^2+n}_{i=d^2+1}M_1(i)$. Clearly, $\overline{M}_1(i)\geq 0~\forall i$ and $\sum^{d^2}_{i=1}\overline{M}_1(i)=\mathbbm{1}_{\cK_1}$. Therefore, $\overline{M}_1=\{\overline{M}_1(1),\ldots,\overline{M}_1(d^2)\}\in\mathscr{M}(\cK_1)$ is a valid POVM. Then from linearity of $\Lambda_1$ and $\Lambda_2$ along with Eq. \eqref{Eq:def_chan_power}, for all $\rho\in\cL(\cH)$ we have 
    \begin{equation}
    \tr[\Lambda_2(\rho)\overline{M}_2(i)]=\tr[\Lambda_1(\rho)\overline{M}_1(i)]~\forall i\in\Omega_{\overline{M}_2}.
    \end{equation} 

    Conversely, if $\overline{M}_2$ is minimal informationally complete in Eq. \eqref{Eq:prop_min_IC_chan_power}, then by definition it is informationally complete and therefore, from Definition \ref{Def:chan_power_IC}, we have $\Lambda_1\succeq_{asymp}\Lambda_2$.
    Hence, the proposition is proved. 
\end{proof}

Next, we provide an equivalent characterization of the relation $\succeq_{asymp}$ using Hermitian-preserving trace-preserving linear maps in the following theorem.
\begin{theorem}
    A channel $\Lambda_1:\cL(\cH)\rightarrow\cL(\cK_1)$ is at least as powerful as $\Lambda_2:\cL(\cH)\rightarrow\cL(\cK_2)$ in the asymptotic sense (i.e., $\Lambda_1\succeq_{asymp}\Lambda_2$) if and only if there exists a Hermitian-preserving trace-preserving linear map $\Theta:\cL(\cK_1)\rightarrow\cL(\cK_2)$ such that $\Lambda_2=\Theta\circ\Lambda_1$.\label{Th:herm_ordering}
\end{theorem}

\begin{proof}
    From Proposition \ref{Prop:chan_power_min_IC}, we have a minimal informationally complete POVM $\overline{M}_2$ and an arbitrary POVM $\overline{M}_1$ such that Eq. \eqref{Eq:prop_min_IC_chan_power} holds. In the Heisenberg picture, we have (similar to Eq. \eqref{Eq:def_chan_power_HP}),

\begin{equation}
    \Lambda_2^{\dagger}(\overline{M}_2(x))=\Lambda_1^{\dagger}(\overline{M}_1(x))~\forall x\in\Omega_{\overline{M}_2}.
\end{equation} 

Note that the set $\{\overline{M}_2(x)\}^{d^2}_{x=1}$ forms a basis of $\cL(\cK_2)$ and consider a linear map $\Theta:\cL(\cK_1)\rightarrow\cL(\cK_2)$ such that 

\begin{align}
    \Theta^{\dagger}[\overline{M}_2(x)]:=\overline{M}_1(x)~\forall x\in\Omega_{\overline{M}_2}.\label{Eq:theta_dagger_define}
\end{align}

Therefore, for an arbitrary matrix $X\in\cL(\cK_2)$ such that $X:=\sum^{d^2}_{x=1}c_x\overline{M}_2(x)$,
\begin{equation}
    \Theta^{\dagger}(X)=\Theta^{\dagger}(\sum^{d^2}_{x=1}c_x\overline{M}_2(x))=\sum^{d^2}_{x=1}c_x\Theta^{\dagger}(\overline{M}_2(x))=\sum^{d^2}_{x=1}c_x\overline{M}_1(x).
\end{equation}
Therefore, 
\begin{align}
    \Lambda_2^{\dagger}(X)&=\Lambda_2^{\dagger}(\sum^{d^2}_{x=1}c_x\overline{M}_2(x))=\sum^{d^2}_{x=1}c_x\Lambda_2^{\dagger}(\overline{M}_2(x))\nonumber\\
    &=\sum^{d^2}_{x=1}c_x\Lambda_1^{\dagger}(\overline{M}_1(x))=\Lambda_1^{\dagger}(\sum^{d^2}_{x=1}c_x\overline{M}_1(x))\nonumber\\
    &=\Lambda_1^{\dagger}(\Theta^{\dagger}(X)).
\end{align}
Hence, $\Lambda^{\dagger}_2=\Lambda^{\dagger}_1\circ\Theta^{\dagger}$ or equivalently $\Lambda_2=\Theta\circ\Lambda_1$. Now, we have to show that $\Theta^{\dagger}$ is a Hermitian-preserving unital map. Note that 

\begin{align}
    \Theta^{\dagger}[X^{\dagger}]&=\Theta^{\dagger}[\sum^{d^2}_{x=1}c^{*}_x\overline{M}^{\dagger}_2(x)]=\sum^{d^2}_{x=1}c^{*}_x\Theta^{\dagger}(\overline{M}^{\dagger}_2(x))\nonumber\\
    &=\sum^{d^2}_{x=1}c^{*}_x\Theta^{\dagger}(\overline{M}_2(x))
    =\sum^{d^2}_{x=1}c^{*}_x\overline{M}_1(x)\nonumber\\
    &=\sum^{d^2}_{x=1}c^{*}_x\overline{M}^{\dagger}_1(x)=[\sum^{d^2}_{x=1}c_x\overline{M}_1(x)]^{\dagger}=[\Theta^{\dagger}(X)]^{\dagger}.
\end{align}

Furthermore, we have

\begin{align}
    \Theta^{\dagger}(\mathbbm{1}_{\cK_2})=\Theta^{\dagger}(\sum^{d^2}_{x=1}\overline{M}_2(x))=\sum^{d^2}_{x=1}\overline{M}_1(x)=\mathbbm{1}_{\cK_1}.
\end{align}

Hence, $\Theta^{\dagger}$ is a Hermitian-preserving unital linear map and therefore, $\Theta$ is a Hermitian-preserving trace-preserving linear map.

Conversely, suppose $\Lambda_2=\Theta\circ\Lambda_1$ where $\Theta$ is a Hermitian-preserving trace-preserving linear map. Then $\Lambda_2^{\dagger}=\Lambda_1^{\dagger}\circ\Theta^{\dagger}$. Therefore, for every $X\in\cL(\cK_2)$, there exists a $Y=\Theta^{\dagger}(X)\in\cL(\cK_1)$ such that $\Lambda_2^{\dagger}(X)=\Lambda_1^{\dagger}(Y)$ holds. Now, consider an arbitrary minimal informationally complete POVM $M^{\prime}_2=\{M^{\prime}_2(x)\}_{x=1}^{d^2}\in\mathscr{M}(\cK_2)$  where $d=\dim(\cK_2)$. Then there exists a set of matrices $M^{\prime}_1=\{M^{\prime}_1(x)=\Theta^{\dagger}(M^{\prime}_2(x))\}$ such that 
    \begin{equation}
    \Lambda_2^{\dagger}(M^{\prime}_2(x))=\Lambda_1^{\dagger}(M^{\prime}_1(x))~\forall x ,\label{Eq:asymp_hp_only_if}
    \end{equation}
    holds. We can see that $M^{\prime}_1(x)=(M^{\prime}_1(x))^{\dagger}~\forall x$ as $\Theta^{\dagger}$ is Hermitian-preserving and $M^{\prime}_1(x)=\Theta^{\dagger}(M^{\prime}_2(x))~\forall x$. Also, $\sum_xM^{\prime}_1(x)=\sum_x\Theta^{\dagger}(M^{\prime}_2(x))=\Theta^{\dagger}(\mathbbm{1}_{\cK_2})=\mathbbm{1}_{\cK_1}$ as $\Theta^{\dagger}$ is unital. But all $M^{\prime}_1(x)$s may not be positive-semidefinite. Now, for an arbitrary Hermitian matrix $A\in\cL(\cK_1)$, let $r_{A}:=\min\{\lambda\in\mathbbm{R}:\lambda\geq 0,0\leq A+\lambda\mathbbm{1}_{\cK_1}\}$. Let $r_{M^{\prime}_1}:=\max_{x}r_{M^{\prime}_1(x)}$. Clearly, we have $0 \leq r_{M^{\prime}_1}<\infty$. Now, consider the set of matrices $M_1=\{M_1(x)=\frac{M^{\prime}_1(x)+r_{M^{\prime}_1}\mathbbm{1}_{\cK_1}}{1+d^2r_{M^{\prime}_1}}\}$. Clearly, $M_1(x)\geq 0~\forall x$ and $\sum_xM_1(x)=\mathbbm{1}_{\cK_1}$. Hence, $M_1$ is a valid POVM. Consider another set of matrices $M_2=\{M_2(x)=\frac{M^{\prime}_2(x)+r_{M^{\prime}_1}\mathbbm{1}_{\cK_2}}{1+d^2r_{M^{\prime}_1}}\}$. Clearly, $M_2(x)\geq 0~\forall x$ and $\sum_xM_2(x)=\mathbbm{1}_{\cK_2}$.  Hence, $M_2$ is also a valid POVM. Then from the unitality and linearity of both $\Lambda^{\dagger}_1$ and $\Lambda^{\dagger}_2$ along with Eq. \eqref{Eq:asymp_hp_only_if} we have 
    \begin{equation}
    \Lambda_2^{\dagger}(M_2(x))=\Lambda_1^{\dagger}(M_1(x))~\forall x. 
    \end{equation}

    Lastly, we have to prove that $M_2$ is an informationally complete POVM. It is sufficient to show that the set of matrices $\{M_2(x)\}$ is linearly independent. Note that as $M^{\prime}_2$ is a minimal informationally complete POVM the set $\{M^{\prime}_2(x)\}$ is linearly independent. For notational simplicity let $t:=\frac{1}{1+d^2r_{M^{\prime}_1}}$ and therefore, $\frac{r_{M^{\prime}_1}}{1+d^2r_{M^{\prime}_1}}=\frac{(1-t)}{d^2}$. Clearly, $0<t\leq 1$. Consider a set of complex numbers $\{\alpha_1,\ldots,\alpha_{d^2}\}$ such that $\sum_x\alpha_xM_2(x)=0$ holds and let $\alpha:=\sum_x\alpha_x$. Then

    \begin{align}
       &\sum_x\alpha_xM_2(x)=0\nonumber\\
       \Rightarrow&\sum_x\alpha_x\Big(tM^{\prime}_2(x)+\frac{(1-t)}{d^2}\mathbbm{1}_{\cK_2}\Big)=0\nonumber\\
       \Rightarrow&\sum_x\alpha_x(tM^{\prime}_2(x))+\frac{(1-t)\alpha}{d^2}\sum_xM^{\prime}_2(x)=0\nonumber\\
       \Rightarrow&\sum_x\Big(\alpha_xt+\frac{(1-t)\alpha}{d^2}\Big)M^{\prime}_2(x)=0\nonumber\\
       \Rightarrow&\alpha_xt+\frac{(1-t)\alpha}{d^2}=0~\forall x\nonumber\\
       \Rightarrow&\alpha_x=-\frac{(1-t)\alpha}{td^2}=:\beta~\forall x,
   \end{align}
   where in the fifth line, we have used the fact that the set $\{M^{\prime}_2(x)\}$ is linearly independent. Then
   \begin{align}
       &\sum_x\alpha_xM_2(x)=0\nonumber\\
       \Rightarrow&\beta\sum_xM_2(x)=0\nonumber\\
       \Rightarrow&\beta\mathbbm{1}_{\cK_2}=0\nonumber\\
       \Rightarrow&\beta=0.
   \end{align}
   Hence, $\alpha_1=\ldots=\alpha_{d^2}=\beta=0$. Hence, $M_2$ is a minimal informationally complete POVM. Therefore, from Eq. \eqref{Eq:def_chan_power_HP}, we have $\Lambda_1\succeq_{asymp}\Lambda_2$. Hence, the theorem is proved.
\end{proof}

\begin{remark}
 \rm{ From Theorem \ref{Th:herm_ordering}, we can conclude that $\Lambda_1\succeq_{asymp}\Lambda_2\Leftrightarrow\Lambda_1\succeq_{HP}\Lambda_2$. Hence, the relation $\succeq_{asymp}$ is a preorder. As quantum channels are Hermitian-preserving trace-preserving linear maps by definition, the relation $\Lambda_1\succeq_{postproc}\Lambda_2\Rightarrow\Lambda_1\succeq_{asymp}\Lambda_2$. Note that Definition \ref{Def:chan_power_IC} is \emph{operationally motivated} (see Remark \ref{Rem:just_asymp} ), and Theorem \ref{Th:herm_ordering} demonstrates its mathematical relation with HPTP maps.  Furthermore, we would like to mention that whenever $\Lambda_1\succeq_{asymp}\Lambda_2$ holds, there may exist multiple Hermitian-preserving trace-preserving linear maps such that $\Lambda_2$ can be written as a composition of these HPTP linear maps with $\Lambda_1$. In other words, it is easy to see that the HPTP linear map $\Theta$ in Theorem \ref{Th:herm_ordering} \emph{may not always} be unique and \emph{also may not be injective}. We would also like to mention that a very special case of Theorem \ref{Th:herm_ordering}, where $\Lambda_2$ is an identity channel, has been discussed in Sec. 2.2
of Ref. \cite{Jiang_2021_physical}. Additionally, our approach is different from the approach in Ref. \cite{Jiang_2021_physical} as informationally complete POVMs have not been discussed in that work. Furthermore, in our approach, we did \emph{not} use the notion of quantum statistical morphism (discussed in Refs. \cite{Buscemi_2018_rev_data, Buscemi_2012_comparison, Buscemi_2014_game, Buscemi_2016_divisibility, Buscemi_2016_degradable}), where unlike Definition \ref{Def:chan_power_IC},  the Eq. \eqref{Eq:def_chan_power} is valid for any POVM i.e. not only for a single informationally complete POVM $M_2\in\cM(\cK_2)$. Also, the notion of channel ordering constructed using a quantum statistical morphism in Refs. \cite{Buscemi_2018_rev_data, Buscemi_2012_comparison, Buscemi_2014_game, Buscemi_2016_divisibility, Buscemi_2016_degradable} is different from our preorder $\succeq_{asymp}$. This can be easily seen from the fact that, according to the former channel ordering, the depolarizing channel and the identity channel are not equivalent, but we will discuss in Example \ref{Examp:depol_iden} that they are equivalent according to our preorder $\succeq_{asymp}$, and this equivalence is operationally meaningful through the communication task depicted in Fig. \ref{Fig:In_retrive_game}. It is also easy to see from this fact that, although the former ordering implies the preordering $\succeq_{asymp}$, the converse is not true, and hence, the preorder relation $\succeq_{asymp}$ is \emph{strictly weaker} than the former ordering constructed using quantum statistical morphisms.} \label{Rm:Mult_HPTP_maps}\end{remark}

Next, we give an example of such channels.

\begin{example}
    \rm{Consider the identity channel $\Lambda_1=\mathbbm{I}_{\cH}$ and a completely dephasing quantum channel $\Lambda_2:\cL(\cH)\rightarrow\cL(\cH)$ such that for an arbitrary quantum state $\rho\in\cS(\cH)$, $\Lambda_2(\rho)=\sum_{i}\ket{i}\bra{i}\rho\ket{i}\bra{i}$ where $\{\ket{i}\}$ forms an orthonormal basis of $\cH$. Note that as $\Lambda_2=\Lambda_2\circ\Lambda_1$, we have $\Lambda_1\succeq_{postproc}\Lambda_2$ and therefore, we have  $\Lambda_1\succeq_{asymp}\Lambda_2$. Now, we show that $\Lambda_2\not\succeq_{asymp}\Lambda_1$. Consider an arbitrary POVM $M=\{M(x)\}\in\mathscr{M}(\cH)$. Then we have
    \begin{align}
        \Lambda_2^{\dagger}(M(x))&=\sum_i\ket{i}\bra{i}M(x)\ket{i}\bra{i}~\forall x\nonumber\\
        &=\Lambda_1^{\dagger}[\sum_i\ket{i}\bra{i}M(x)\ket{i}\bra{i}]~\forall x\nonumber\\
        &=\Lambda_1^{\dagger}[M^{\prime}(x)]~\forall x,
    \end{align}
    where $M^{\prime}=\{M^{\prime}(x)=\sum_i\ket{i}\bra{i}M(x)\ket{i}\bra{i}\}$ is a valid POVM. Note that all POVM elements are diagonal in $\{\ket{i}\}$ basis and therefore, it cannot be an informationally complete POVM (as the matrices $\ket{i}\bra{j}$ for $i\neq j$ cannot be written as a linear combination of only diagonal matrices). Hence, in this case, there is no informationally complete POVM such that Eq. \eqref{Eq:def_chan_power} is satisfied. Therefore, $\Lambda_2\not\succeq_{asymp}\Lambda_1$.
    }
\end{example}

Now, an important question is whether $\Lambda_1\succeq_{postproc}\Lambda_2$ holds given that $\Lambda_1\succeq_{asymp}\Lambda_2$ holds. We show the negative result through an example. But before that we prove the following useful Lemma.

\begin{lemma}
    If the POVM $M=\{M(x)\}^{d^2}_{x=1}\in\mathscr{M}(\cK_1)$ is a minimal informationally complete POVM then the POVM $M^{\prime}=\{M^{\prime}(x)=tM(x)+(1-t)\frac{\tr[M(x)]}{d}\mathbbm{1}_{\cK_1}\}^{d^2}_{x=1}\in\mathscr{M}(\cK_1)$ is also informationally complete where $0<t\leq 1$ and $d=\dim(\cK_1)$. \label{Lem:noisy_IC}
\end{lemma}

\begin{proof}
   Clearly, it suffices to show that the set $\{M^{\prime}(x)\}^{d^2}_{x=1}$ is linearly independent. Consider a set of complex numbers $\{\alpha_x\}^{d^2}_{x=1}$. Then 
   \begin{align}
       &\sum_x\alpha_xM^{\prime}(x)=0\nonumber\\
       \Rightarrow&\sum_x\alpha_x\Big(tM(x)+(1-t)\frac{\tr[M(x)]}{d}\mathbbm{1}_{\cK_1}\Big)=0\nonumber\\
       \Rightarrow&\sum_x\alpha_xtM(x)+(1-t)\sum_y\alpha_y\frac{\tr[M(y)]}{d}\mathbbm{1}_{\cK_1}=0\nonumber\\
       \Rightarrow&\sum_x\alpha_xtM(x)+(1-t)\sum_y\alpha_y\lambda_y\mathbbm{1}_{\cK_1}=0,\label{Eq:lem_noisy_ic}
   \end{align}
   where $\lambda_y=\frac{\tr[M(y)]}{d}\forall~y$. Therefore, from Eq. \eqref{Eq:lem_noisy_ic} we have
   \begin{align}
       &\sum_x\alpha_xtM(x)+\sum_x(1-t)(\sum_y\alpha_y\lambda_y)M(x)=0\nonumber\\
       \Rightarrow&\sum_x(\alpha_xt+(1-t)(\sum_y\alpha_y\lambda_y))M(x)=0.\label{Eq:lem_noisy_ic_alpha}
   \end{align}
   As the set $\{M(x)\}$ is linearly independent, from Eq. \eqref{Eq:lem_noisy_ic_alpha}, we have
   \begin{align}
       &\alpha_xt+(1-t)(\sum_y\alpha_y\lambda_y)=0~\forall x\nonumber\\
       \Rightarrow&\alpha_x=-\frac{(1-t)}{t}(\sum_y\alpha_y\lambda_y)~\forall x.\label{Eq:lem_noisy_ic_lin_in}
   \end{align}
   Hence, we have $\alpha_1=\alpha_2=\ldots=\alpha_{d^2}=-\frac{(1-t)}{t}(\sum_y\alpha_y\lambda_y)=:\alpha$ (say).
   Then we can write

\begin{align}
\sum_x\alpha_x M^{\prime}(x)=&\alpha\sum_x M^{\prime}(x)\nonumber\\
       =&\alpha \mathbbm{1}_{\cK_1}=0\nonumber\\\Rightarrow&\alpha=\alpha_x=0\quad\forall x.
\end{align}
  
   Hence, the set $\{M^{\prime}(x)\}^{d^2}_{x=1}$ is linearly independent.
\end{proof}
Now, we are ready to discuss the example.
\begin{example}
    \rm{Consider a depolarizing quantum channel $\Lambda_1:\cL(\cH)\rightarrow\cL(\cH)$ such that for an arbitrary quantum state $\rho\in\cS(\cH)$, $\Lambda_1(\rho)=t\rho+(1-t)\frac{\mathbbm{1}_{\cH}}{d}$ with $0<t< 1$ and the identity channel $\Lambda_2=\mathbbm{I}_{\cH}$. At first, we show $\Lambda_1\simeq_{asymp}\Lambda_2$. Consider an informationally complete POVM $M=\{M(x)\}^{d^2}_{x=1}\in\mathscr{M}(\cH)$.
    From Lemma \ref{Lem:noisy_IC}, we obtain that the POVM $\overline{M}=\Lambda^{\dagger}_1(M)=\{tM(x)+(1-t)\frac{\tr(M(x))}{d}\mathbbm{1}_{\cH}\}$ is informationally complete. Then we have
    \begin{align}
        \tr[\Lambda_1(\rho)M(x)]&=\tr[\rho\Lambda_1^{\dagger}(M(x))]~\forall x\nonumber\\
        &=\tr[\Lambda_2(\rho)\overline{M}(x)]~\forall x.\label{Eq:examp_dep_id_chan_asymp_eq}
    \end{align}
    Noting that both POVMs $M$ and $\overline{M}$ are informationally complete and from Eq. \eqref{Eq:def_chan_power} and Eq. \eqref{Eq:examp_dep_id_chan_asymp_eq}, we have both $\Lambda_1\succeq_{asymp}\Lambda_2$ and $\Lambda_2\succeq_{asymp}\Lambda_1$ and therefore, $\Lambda_1\simeq_{asymp}\Lambda_2$. Note that in our case, $\Lambda_1\circ\Lambda_2=\Lambda_1$ and therefore, $\Lambda_2\succeq_{postproc}\Lambda_1$. But $\Lambda_1\not\succeq_{postproc}\Lambda_2$. This is because if $\Lambda_1\succeq_{postproc}\Lambda_2$ then there exists a quantum channel $\Gamma:\cL(\cH)\rightarrow\cL(\cH)$ such that $\Lambda_2=\Gamma\circ\Lambda_1$. Now, note that for an arbitrary pair of orthogonal pure states $\psi=\ket{\psi}\bra{\psi}$ and $\phi=\ket{\phi}\bra{\phi}$, we have $||\Lambda_2(\psi)-\Lambda_2(\phi)||_1>||\Lambda_1(\psi)-\Lambda_1(\phi)||_1$. But we know that a quantum channel $\Gamma$ cannot increase trace distance. Hence, $\Lambda_1\not\succeq_{postproc}\Lambda_2$. Summarizing, we obtain that although  $\Lambda_1\succeq_{asymp}\Lambda_2$ holds, we have $\Lambda_1\not\succeq_{postproc}\Lambda_2$. Furthermore, as $\Lambda_1\succeq_{asymp}\Lambda_2$ holds, from Theorem \ref{Th:herm_ordering}, we obtain that there exists a Hermitian-preserving trace-preserving linear map $\Theta$ such that $\Lambda_2=\Theta\circ\Lambda_1$. Now, again as we know that $||\Lambda_2(\psi)-\Lambda_2(\phi)||_1>||\Lambda_1(\psi)-\Lambda_1(\phi)||_1$, and positive trace-preserving maps cannot increase trace distance \cite{watrous_2018_theory}, we conclude that the map $\Theta$ is not positive. Note that the more $t$ is close to zero, the more it is difficult to reconstruct the quantum state $\Lambda_2(\rho)=\rho$ from $\Lambda_1(\rho)$ as it is closer to the maximally mixed state. But we are concerned with only the possibility of reconstructing the state $\Lambda_2(\rho)$ from $\Lambda_1(\rho)$ (i.e., to perfectly accomplish the task depicted in Fig. \ref{Fig:In_retrive_game}). Hence, the identity channel $\Lambda_2$ and the depolarizing channel $\Lambda_1$ are equivalent in the sense that, for an given unknown input state, the output state of either of them can be uniquely determined from the output of the other.}  \label{Examp:depol_iden}
\end{example}
From Example \ref{Examp:depol_iden}, we make the following simple observations.

\begin{obs}
    \rm{Given two arbitrary quantum channels $\Lambda_1$ and $\Lambda_2$
    \begin{enumerate}
        \item $\Lambda_1\succeq_{asymp}\Lambda_2$ does not imply $\Lambda_1\succeq_{postproc}\Lambda_2$, in general.\label{obs:coun_exp_1}
        \item If $\Lambda_1\succeq_{asymp}\Lambda_2$, the Hermitian-preserving trace-preserving map $\Theta$ in Theorem \ref{Th:herm_ordering} may not always be positive.\label{obs:coun_exp_2}
    \end{enumerate}
        }\label{obs:coun_exp}
\end{obs}



    
Now, if a quantum channel $\Lambda_1\succeq_{asymp}\Lambda_2$ from Remark \ref{Rem:just_asymp}, we know that with a large number of copies of the quantum state $\Lambda_1(\rho)$ the quantum state $\Lambda_2(\rho)$ can be determined for an arbitrary unknown state $\rho$. So, the general intuition suggests that from Eq. \eqref{Eq:def_chan_power} if we know about the incompatibility of the minimal informationally complete POVM $M_2$ and the channel $\Lambda_1$, incompatibility between $\Lambda_1$ and $\Lambda_2$ can be determined. But we show in the following, by an explicit example, that it is not the case.

\begin{example}[Implication on quantum incompatibility]\rm{Let $\Lambda_1$ be the depolarizing channel and $\Lambda_2$ be the identity channel defined in Example \ref{Examp:depol_iden}. Then $\overline{\Lambda}_1$ represents the conjugate channel of $\Lambda_1$. Next, consider two quantum channels, $\Phi_1=\overline{\Lambda}_1$ and $\Phi_2=\Lambda_2$. Also consider an informationally complete POVM $M=\{M(x)\}^{d^2}_{x=1}\in\mathscr{M}(\cH)$. As $\Phi_2$ is the identity channel, from Lemma \ref{Lem:noisy_IC}, we know that
\begin{align}
    \Phi_2^{\dagger}(\overline{M}(x))=\overline{M}(x)=\overline{\Phi}_1^{\dagger}(M(x))=\Lambda_1^{\dagger}(M(x))\quad\forall x, \label{Eq:meas_chan_incomp} 
\end{align}
is also an informationally complete POVM. From Proposition \ref{Prop:meas_chan_incom}, we get that Eq. \eqref{Eq:meas_chan_incomp} is equivalent to the fact that the informationally complete POVM $\Phi_2^{\dagger}(\overline{M})$ and the quantum channel $\Phi_1$ are compatible. With all of these definitions, our goal is to show that the compatibility of $\Phi_2^{\dagger}(\overline{M})$ and $\Phi_1$ can't be used to determine the compatibility of the quantum channels $\Phi_1$ and $\Phi_2$ even though $\overline{M}$ is informationally complete.

 By Definition \ref{Def:chan_power_IC} and Eq. \eqref{Eq:def_chan_power_HP}, we see that Eq. \eqref{Eq:meas_chan_incomp} is equivalent to
\begin{align}
    \overline{\Phi}_1\succeq_{asymp}\Phi_2\quad\text{or}\quad\Lambda_1\succeq_{asymp}\Lambda_2.
\end{align}


But from Example \ref{Examp:depol_iden} we already know that $\Lambda_1\not\succeq_{postproc}\Lambda_2$ (also see statement \ref{obs:coun_exp_1} of Observation \ref{obs:coun_exp}). Thus from the condition \ref{Prop:chan_chan_incomp_point_3} in Proposition \ref{Prop:chan_chan_incomp}, we can conclude that $\Phi_1$ and $\Phi_2$ are not compatible.} \label{Examp:incomp_asymp}
\end{example}

\begin{remark}
\rm{Example \ref{Examp:incomp_asymp} may sound simple to some readers. But our goal was to illustrate that the statement ''the incompatibility of channels cannot be cast in terms of measurement-channel compatibility, in general'' can be understood as a consequence of statement \ref{obs:coun_exp_1} of Observation \ref{obs:coun_exp}.}
\end{remark}

Similarly, from the statement \ref{obs:coun_exp_1} of Observation \ref{obs:coun_exp}, one can argue that measurement-channel compatibility cannot be cast in terms of compatibility of POVMs.

\begin{theorem}
    Given two quantum channels $\Lambda_1\in\mathscr{C}(\cH,\cK_1)$ and $\Lambda_2\in\mathscr{C}(\cH,\cK_2)$, the relation $\Lambda_1\succeq_{asymp}\Lambda_2$ is satisfied if and only if $\ker(\Lambda_1)\subseteq \ker(\Lambda_2)$.\label{Th:kern_ordering}
\end{theorem}
\begin{proof}
    If $\Lambda_1\succeq_{asymp}\Lambda_2$, then equivalently, we have
    \begin{align}
        \tr[\Lambda_2(\rho)M_2(x)]=\tr[\Lambda_1(\rho)M_1(x)],
    \end{align}

    for every $\rho\in\cL(\cH)$. Here, $M_2\in\mathscr{M}(\cK_2)$ is a minimal informationally complete POVM and $M_1\in\mathscr{M}(\cK_1)$ is another POVM.
Let us first assume $\rho\in \ker(\Lambda_1)$ \textit{i.e.,}
\begin{align}
    \Lambda_1(\rho)=0.
\end{align}
This implies that
\begin{align}
    \tr[\Lambda_1(\rho)M_1(x)]=\tr[\Lambda_2(\rho)M_2(x)]=0 \quad\forall~x.
\end{align}
As $M_2$ is informationally complete
\begin{align}
    \tr[\Lambda_2(\rho)M_2(x)]=0~\forall x\Rightarrow~\Lambda_2(\rho)=0.
\end{align}
Hence,
\begin{align}
    \rho\in \ker(\Lambda_1)\Rightarrow \rho\in \ker(\Lambda_2).\label{Eq:ker1_imply_ker2}
\end{align}

 Thus, we can conclude that
\begin{align}
    \ker(\Lambda_1)\subseteq \ker(\Lambda_2).\label{Eq:kernal_subset}
\end{align}

Conversely, let $ \ker(\Lambda_1)\subseteq \ker(\Lambda_2)$. Then, we can define a map $\Theta_1:Im(\Lambda_1)\rightarrow \cL(\cK_2)$ such that
\begin{align} \Theta_1(\Lambda_1(\rho)):=\Lambda_2(\rho)\quad\forall~\rho\in\cL(\cH).
\end{align}
Here, the symbol $Im$ refers to the image of a linear map. Consider $\rho,\sigma\in\cL(\cH)$ such that $\Lambda_1(\rho)=\Lambda_1(\sigma)$. This implies
\begin{align}
    &\Lambda_1(\rho-\sigma)=0\nonumber\\
    \Rightarrow& \rho-\sigma\in \ker(\Lambda_1).
\end{align}

As $\ker(\Lambda_1)\subseteq \ker(\Lambda_2)$ we have
\begin{align}
    &\rho-\sigma\in \ker(\Lambda_2)\nonumber\\\Rightarrow&\Lambda_2(\rho)=\Lambda_2(\sigma).
    \end{align}
Thus $\Lambda_1(\rho)=\Lambda_1(\sigma)\Rightarrow \Theta_1(\Lambda_1(\rho))=\Theta_1(\Lambda_1(\sigma))$. Hence, the map $\Theta_1$ is well-defined. 

For $\alpha\Lambda_1(\rho)+\beta\Lambda_1(\sigma)\in Im(\Lambda_1)$ where $\alpha,\beta\in\mathbbm{C}$, we have
\begin{align}
\Theta_1(\alpha\Lambda_1(\rho)+\beta\Lambda_1(\sigma))=&\Theta_1(\Lambda_1(\alpha\rho+\beta\sigma))\nonumber\\
=&\Lambda_2(\alpha\rho+\beta\sigma)\nonumber\\
=&\alpha\Lambda_2(\rho)+\beta\Lambda_2(\sigma)\nonumber\\
=&\alpha\Theta_1(\Lambda_1(\rho))+\beta\Theta_1(\Lambda_1(\sigma)).
\end{align}
Thus, $\Theta_1$ is a linear map.

Next, as we have
\begin{align}
\Theta_1((\Lambda_1(\rho))^{\dagger})&=\Theta_1(\Lambda_1(\rho^{\dagger}))\nonumber\\
&=\Lambda_2(\rho^{\dagger})\nonumber\\
&=(\Lambda_2(\rho)
)^{\dagger}=(\Theta_1(\Lambda_1(\rho)))^{\dagger}.\end{align}
The linear map $\Theta_1$ is Hermitian-preserving. 

Also
\begin{align}
    \tr[\Theta_1(\Lambda_1(\rho))]=&\tr[\Lambda_2(\rho)]\nonumber\\
    =&\tr[\rho]=\tr[\Lambda_1(\rho)].
\end{align}
Thus, the map $\Theta_1$ is trace-preserving also.

Next, define another map $\Theta_2:Im(\Lambda_1)^{\perp}\rightarrow \cL(\cK_2)$ such that for all $Y\in Im(\Lambda_1)^{\perp}$
\begin{align}
    \Theta_2(Y):=\tr[Y]\frac{\mathbbm{1}_{\cK_2}}{\dim(\cK_2)}.
\end{align}
We can see that it is linear, as for any $\alpha Y_1+\beta Y_2\in Im(\Lambda_1)^{\perp}$ we have
\begin{align}
    \Theta_2(\alpha Y_1+\beta Y_2)=\alpha\Theta_2(Y_1)+\beta \Theta_2(Y_2).
\end{align}

Also as 
\begin{align}
\Theta_2(Y^\dagger)=&(\tr[Y])^*\frac{\mathbbm{1}_{\cK_2}}{\dim(\cK_2)}\nonumber\\
    =&(\Theta_2(Y))^{\dagger},
\end{align}
the linear map $\Theta_2$ is Hermitian-preserving. 

It is also trace-preserving as
\begin{align}
    \tr[\Theta_2(Y)]=\tr[Y].
\end{align}

Now, as we know
\begin{align}
    \cL(\cK_1)=Im(\Lambda_1)\oplus Im(\Lambda_1)^{\perp},
\end{align}
an arbitrary $Z\in\cL(\cK_1)$ can be written in a unique decomposition as
\begin{align}
    Z=X\oplus 0_{Im(\Lambda_1)^{\perp}}+0_{Im(\Lambda_1)}\oplus Y,\label{Eq:Z_form_direct_sum}
\end{align}
where $X\in Im(\Lambda_1)$ and $Y\in Im(\Lambda_1)^{\perp}$. Here, $0_{Im(\Lambda_1)^{\perp}}$ is the null matrix in the vector space $Im(\Lambda_1)^{\perp}$ and $0_{Im(\Lambda_1)}$ is the null matrix in the vector space $Im(\Lambda_1)$. Let us define another map $\Theta:\cL(\cK_1)\rightarrow\cL(\cK_2)$ as
\begin{align}
    \Theta(Z):=\Theta_1(X)+\Theta_2(Y).
\end{align}
For $\alpha Z_1+\beta Z_2\in \cL(\cK_1)$, where $Z_1=X_1\oplus 0_{Im(\Lambda_1)^{\perp}}+0_{Im(\Lambda_1)}\oplus Y_1$ and $Z_2=X_2\oplus 0_{Im(\Lambda_1)^{\perp}}+0_{Im(\Lambda_1)}\oplus Y_2$, we have
\begin{align}
    \Theta(\alpha Z_1+\beta Z_2)=&\Theta_1(\alpha X_1+\beta X_2)+\Theta_2(\alpha Y_1+\beta Y_2),\nonumber\\
    =&\alpha(\Theta_1(X_1)+\Theta_2(Y_1))+\beta(\Theta_1(X_2)+\Theta_2(Y_2)),\nonumber\\
    =&\alpha\Theta(Z_1)+\beta\Theta(Z_2).
\end{align}
Thus $\Theta$ is a linear map.

Clearly, for an arbitrary $Z\in\cL(\cK_1)$ (as in Eq. \eqref{Eq:Z_form_direct_sum}) we have
\begin{align}
    \Theta(Z^{\dagger})=&\Theta_1(X^{\dagger})+\Theta_2(Y^{\dagger}),\nonumber\\
    =&(\Theta_1(X))^{\dagger}+(\Theta_2(Y))^{\dagger}=\Theta(Z)^{\dagger},
\end{align}
and 
\begin{align}
    \tr[\Theta(Z)]=&\tr[\Theta_1(X)]+\tr[\Theta_2(Y)],\nonumber\\
    =& \tr[X]+\tr[Y]=\tr[Z].
\end{align}
Hence, $\Theta$ is a Hermitian-preserving trace-preserving linear map.
Now, as $\Lambda_1(\rho)\in Im(\Lambda_1)$ for an arbitrary $\rho\in\cL(\cH)$, we have
\begin{align}
    \Theta(\Lambda_1(\rho))=\Theta_1(\Lambda_1(\rho))+\Theta_2(0_{Im(\Lambda_1)^{\perp}})=\Lambda_2(\rho).
\end{align}
Hence, we have $\Lambda_2=\Theta\circ\Lambda_1$ which in turn, by Theorem \ref{Th:herm_ordering} implies $\Lambda_1\succeq_{asymp}\Lambda_2$.
\end{proof}

\begin{theorem}
    Consider an arbitrary pair of quantum channels $\Lambda_1\in\mathscr{C}(\cH,\cK_1)$ and $\Lambda_2\in\mathscr{C}(\cH,\cK_2)$. The condition $\Lambda_1\succeq_{asymp}\Lambda_2$ holds if and only if the equality $\mid\mid\Lambda_1(\rho)-\Lambda_1(\sigma)\mid\mid_1=0$ implies the equality $\mid\mid\Lambda_2(\rho)-\Lambda_2(\sigma)\mid\mid_1=0,~\text{where}~\rho,\sigma\in\cS(\cH)$. \label{Th:asymp_trace_dist}
\end{theorem}

\begin{proof}
    Assume that $\Lambda_1\succeq_{asymp}\Lambda_2$ holds. From Theorem \ref{Th:kern_ordering}, we have $\ker(\Lambda_1)\subseteq\ker(\Lambda_2)$. Then
    \begin{align}
        &\mid\mid\Lambda_1(\rho)-\Lambda_1(\sigma)\mid\mid_1=0,\nonumber\\
        \Rightarrow&\Lambda_1(\rho-\sigma)=0,\nonumber\\
        \Rightarrow &(\rho-\sigma)\in \ker{\Lambda_1},\nonumber\\
        \Rightarrow &(\rho-\sigma)\in \ker{\Lambda_2},\nonumber\\
        \Rightarrow &\Lambda_2(\rho-\sigma)=0,\nonumber\\
        \Rightarrow &\mid\mid\Lambda_2(\rho)-\Lambda_2(\sigma)\mid\mid_1=0.\nonumber\\
    \end{align}

    Conversely, assume that the equality $\mid\mid\Lambda_1(\rho)-\Lambda_1(\sigma)\mid\mid_1=0$ implies the equality $\mid\mid\Lambda_2(\rho)-\Lambda_2(\sigma)\mid\mid_1=0$. Now note that as both $\Lambda_1$ and $\Lambda_2$ are quantum channels, they are also Hermitian- preserving. Therefore, they satisfy the condition $\Lambda_i(X^{\dagger})=\Lambda_i(X)^{\dagger}~\forall X\in\cL(\cH), i=1,2$. Therefore, for an arbitrary $X\in\cL(\cH)$, if $X\in\ker(\Lambda_i)$, then $X^{\dagger}\in\ker(\Lambda_i)~\forall i=1,2$. Therefore, $\frac{X+X^{\dagger}}{2}\in\ker(\Lambda_i)~\forall i=1,2$ and $\frac{X-X^{\dagger}}{2}\in\ker(\Lambda_i)~\forall i=1,2$. Furthermore, as both $\Lambda_1$ and $\Lambda_2$ are trace-preserving, $\tr[X]=0$. Observing these facts, consider a complete basis $\{X_j\}^n_{j=1}$ of $\ker(\Lambda_1)$. Define a set $\{Y_j\}^n_{j=1}$ such that $Y_j=\frac{X+X^{\dagger}}{2}$ $\forall$ $j\in\{1,\ldots,n\}$ and $Y_{n+j}=i\frac{X-X^{\dagger}}{2}$ $\forall$ $j\in\{1,\ldots,n\}$. Hence, $Y_j$ is Hermitian $\forall ~j\in\{1,\ldots,2n\}$. Then, for an arbitrary matrix $A=\sum_jc^{\prime}_jX_j\in\ker(\Lambda_1)$, we have
    \begin{align}
        A&=\sum^{n}_{j=1}c^{\prime}_jX_j,\nonumber\\
        &=\sum^{n}_{j=1}c^{\prime}_j(\frac{X_j+X_j^{\dagger}}{2}+\frac{X_j-X_j^{\dagger}}{2}),\nonumber\\
        &=\sum^{2n}_{j=1}c_jY_j,
    \end{align}
    where $c_j=ic_{n+j}=c^{\prime}_j\forall j\in\{1,\ldots,n\}$. Hence, $\{Y_j\}^{2n}_{j=1}$ forms a overcomplete basis of $\ker(\Lambda_1)$. Furthermore, as $\tr[X_j]=0~\forall j\in\{1,\ldots,n\}$, we have $\tr[Y_j]=0~\forall j\in\{1,\ldots,2n\}$. Now, for an arbitrary $j\in\{1,\ldots,2n\}$, if $Y_j=0$, we trivially have $\Lambda_1(Y_j)=\Lambda_2(Y_j)=0$ and hence, $Y_j\in\ker(\Lambda_2)$. But if $Y_j\neq 0$, we can write $Y_j=Q_j-R_j$ where $Q_j\geq 0$ and $R_j\geq 0$ as $Y_j$ is Hermitian. Additionally, as $\tr[Y_j]=0$, we have both $Q_j,R_j\neq 0$ which in turn implies, $\tr[Q_j]> 0$ and $\tr[R_j]> 0$. Let us define $\tr[Q_j]=\tr[R_j]:=\gamma_j$. Now, consider the density matrices $\rho_j:=\frac{Q_j}{\tr[Q_j]}$ and $\sigma_j:=\frac{R_j}{\tr[R_j]}$. Then

    \begin{align}
        &\Lambda_1(Y_j)=0,\nonumber\\
        \Rightarrow&\gamma_j\Lambda_1(\rho_j-\sigma_j)=0,\nonumber\\
        \Rightarrow&\gamma_j\Lambda_2(\rho_j-\sigma_j)=0,\nonumber\\
        \Rightarrow&\Lambda_2(Q_j-R_j)=0,\nonumber\\
        \Rightarrow&\Lambda_2(Y_j)=0,\nonumber\\
    \end{align}
    holds $\forall$ $j\in\{1,\ldots,2n\}$. Hence, we have $\Lambda_2(A)=0$ and therefore, $A\in\ker(\Lambda_2)$. As $A\in\ker(\Lambda_1)$ is chosen arbitrarily, we have $\ker(\Lambda_1)\subseteq\ker(\Lambda_2)$. Hence, from Theorem \ref{Th:kern_ordering}, we have $\Lambda_1\succeq_{asymp}\Lambda_2$.

\end{proof}

\begin{remark}
    \rm{From Definition \ref{Def:chan_power_IC}, Proposition \ref{Prop:chan_power_min_IC}, Theorem \ref{Th:herm_ordering}, Theorem \ref{Th:kern_ordering}, and Theorem \ref{Th:asymp_trace_dist} we obtain five equivalent definitions of the relation $\succeq_{asymp}$.} 
\end{remark}

\subsection{Hierarchy of the preorders}\label{Sec:Main_res_Hierarchy}
Consider a quantum channel $\Lambda:\cL(\cH)\rightarrow\cL(\cK)$. Let us denote the set of all quantum channels obtained by concatenating it with Hermitian-preserving trace-preserving linear maps as $\mathscr{C}^{HP}_{\cH\rightarrow\overline{\cH}}(\Lambda)$. In other words
\begin{align}
    \mathscr{C}^{HP}_{\cH\rightarrow\overline{\cH}}(\Lambda):=\{\Phi:\Phi\in\mathscr{C}(\cH,\overline{\cH}),\Lambda\succeq_{HP}\Phi\}.
\end{align}
 It is worth mentioning that the set $\mathscr{C}^{HP}_{\cH\rightarrow\overline{\cH}}(\Lambda)$ is a convex set. This can easily be seen by choosing two quantum channels $\Phi_1\in\mathscr{C}^{HP}_{\cH\rightarrow\overline{\cH}}(\Lambda)$ and $\Phi_2\in\mathscr{C}^{HP}_{\cH\rightarrow\overline{\cH}}(\Lambda)$ such that $\Phi_1=\Theta_1\circ\Lambda$ and $\Phi_2=\Theta_2\circ\Lambda$. Here, $\Theta_1,\Theta_2\in\mathscr{C}_{HP}(\cK,\overline{\cH})$. By taking the convex combination of $\Phi_1$ and $\Phi_2$ we can easily see that
 \begin{align}
     \lambda\Phi_1+(1-\lambda)\Phi_2=&\lambda\Theta_1\circ\Lambda+(1-\lambda)\Theta_2\circ\Lambda\nonumber\\
     =&(\lambda\Theta_1+(1-\lambda)\Theta_2)\circ\Lambda\in \mathscr{C}^{HP}_{\cH\rightarrow\overline{\cH}}(\Lambda),
 \end{align}
 for $0\leq\lambda\leq 1$ as $\lambda\Theta_1+(1-\lambda)\Theta_2\in\mathscr{C}_{HP}(\cK,\overline{\cH})$. Hence $\Lambda\succeq_{HP}\lambda\Phi_1+(1-\lambda)\Phi_2$. Thus, $\mathscr{C}^{HP}_{\cH\rightarrow\overline{\cH}}(\Lambda)$ forms a convex set.
 
 From Definition
\ref{Def:chan_power_IC}, we can also define the set of channels $\Lambda$ is at least as powerful as in the asymptotic sense as $\mathscr{C}^{asymp}_{\cH\rightarrow\overline{\cH}}(\Lambda)$ such that
\begin{align}
    \mathscr{C}^{asymp}_{\cH\rightarrow\overline{\cH}}(\Lambda)=\{\Phi:\Phi\in\mathscr{C}(\cH,\overline{\cH}),\Lambda\succeq_{asymp}\Phi\}.
\end{align}
 From Theorem \ref{Th:herm_ordering} we get $\mathscr{C}^{asymp}_{\cH\rightarrow\overline{\cH}}(\Lambda)=\mathscr{C}^{HP}_{\cH\rightarrow\overline{\cH}}(\Lambda)$.

 Similarly, let us denote the set of all quantum channels obtained by concatenating $\Lambda$ with positive maps as $\mathscr{C}^{P}_{\cH\rightarrow\overline{\cH}}(\Lambda)$. In other words,
\begin{align}
        \mathscr{C}^{P}_{\cH\rightarrow\overline{\cH}}(\Lambda)=\{\Phi:\Phi\in\mathscr{C}(\cH,\overline{\cH}),\Lambda\succeq_{P}\Phi\}.
\end{align}
 Just like the previous case, the set $\mathscr{C}^{P}_{\cH\rightarrow\overline{\cH}}(\Lambda)$ is also a convex set. 

For completeness, we also denote the set of all quantum channels obtained by post-processing it with another quantum channel as $\mathscr{C}^{CP}_{\cH\rightarrow\overline{\cH}}(\Lambda)$ such that
\begin{align}
    \mathscr{C}^{CP}_{\cH\rightarrow\overline{\cH}}(\Lambda)=\{\Phi:\Phi\in\mathscr{C}(\cH,\overline{\cH}),\Lambda\succeq_{postproc}\Phi\}.
\end{align}
It is known that the set $\mathscr{C}^{CP}_{\cH\rightarrow\overline{\cH}}(\Lambda)$ is also convex.

It is easy to see that the above four sets of channels satisfy the following subset relations, as illustrated in Fig. \ref{fig_subset_rel}:
\begin{align}
   \mathscr{C}^{CP}_{\cH\rightarrow\overline{\cH}}(\Lambda)\subseteq\mathscr{C}^{P}_{\cH\rightarrow\overline{\cH}}(\Lambda)\subseteq\mathscr{C}^{asymp}_{\cH\rightarrow\overline{\cH}}(\Lambda)=\mathscr{C}^{HP}_{\cH\rightarrow\overline{\cH}}(\Lambda).
\end{align}
The preorders arising from completely positive and positive maps are studied in the form of CP-divisibility and P-divisibility in the literature. For preliminary references and further mathematical details, we refer interested readers to Refs. \cite{WolfCirac2008,Davalos2019,DavalosZiman2023,vomEnde2025} and the references therein.
\begin{figure}
    \centering
    \includegraphics[height=140px, width =256px]{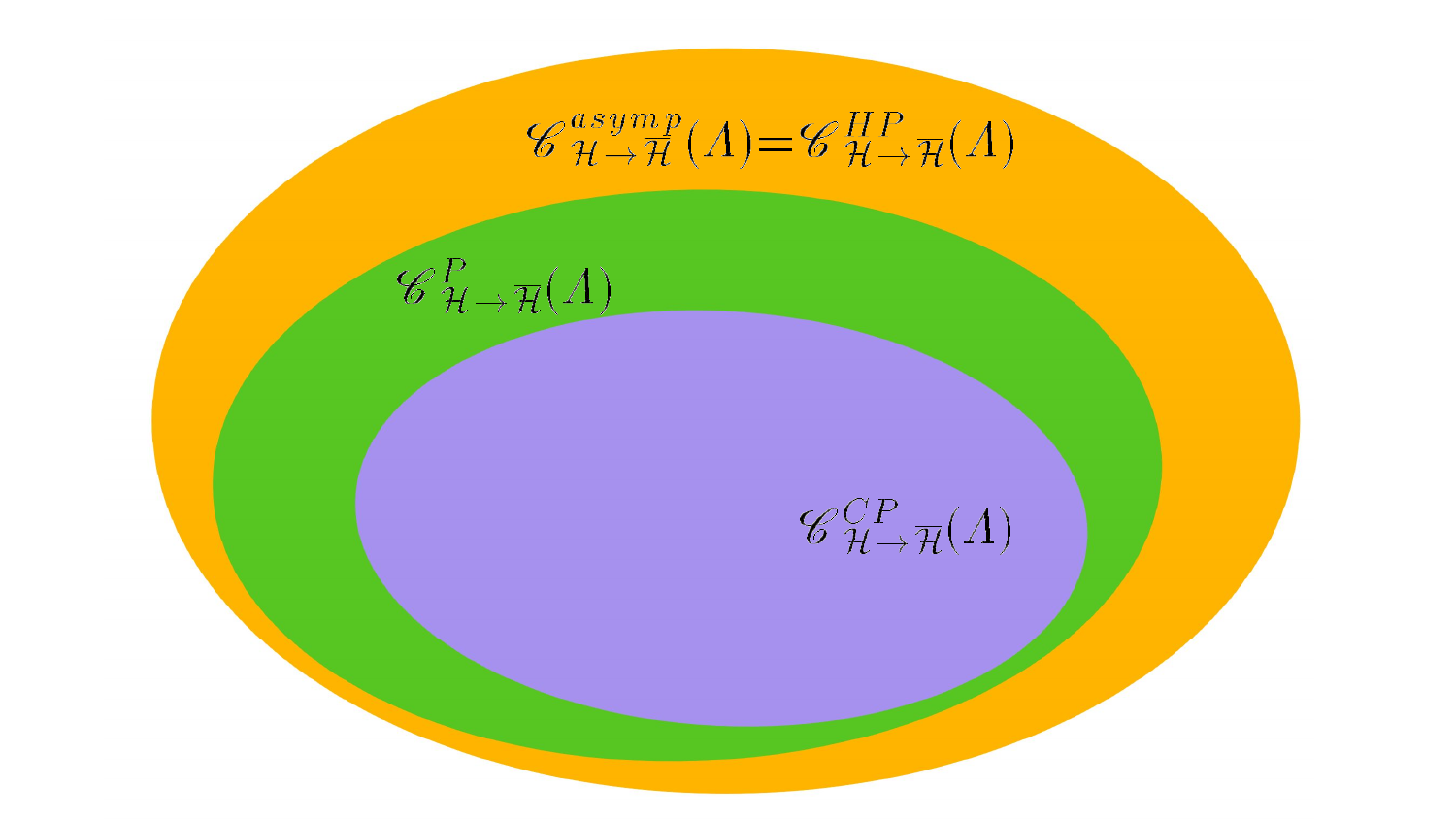}
     \caption{For a given channel $\Lambda\in\mathscr{C}(\cH,\cK)$, $\mathscr{C}^{HP}_{\cH\rightarrow\overline{\cH}}(\Lambda)$ represents the set of all quantum channels obtained by concatenating it with Hermitian-preserving trace-preserving linear maps, $\mathscr{C}^{P}_{\cH\rightarrow\overline{\cH}}(\Lambda)$ represents the set of all quantum channels obtained by concatenating it with positive maps and $\mathscr{C}^{CP}_{\cH\rightarrow\overline{\cH}}(\Lambda)$ represents the set of all quantum channels obtained by post-processing it with quantum channels. These sets satisfy the following subset relations: $\mathscr{C}^{CP}_{\cH\rightarrow\overline{\cH}}(\Lambda)\subseteq\mathscr{C}^{P}_{\cH\rightarrow\overline{\cH}}(\Lambda)\subseteq\mathscr{C}^{asymp}_{\cH\rightarrow\overline{\cH}}(\Lambda)=\mathscr{C}^{HP}_{\cH\rightarrow\overline{\cH}}(\Lambda) $.}\label{fig_subset_rel} 
\end{figure}

\subsection{Physical Implementability}\label{Sec:Main_res_Phys_Imp}
In this section, our goal is to study the physical implementability of $\Lambda_2$ given that $\Lambda_1$ has already been performed where $\Lambda_1\succeq_{asymp}\Lambda_2$.
But first, we discuss a quantifier named \emph{physical implementability} that was introduced in Ref. \cite{Jiang_2021_physical} to characterize the \emph{difficulty} of the physical approximation of a given HPTP linear map. The HPTP map is first decomposed into a linear combination of CPTP maps. The \emph{physical approximation} is then characterized by the expansion coefficients. Thus, it quantifies the cost of simulating a given HPTP linear map using physically implementable quantum
operations. It has an operational interpretation in the framework of quantum error mitigation, as it provides a lower bound on the sampling cost\cite{Jiang_2021_physical} achievable using the quasiprobability decomposition technique \cite{CROWDER2015157,PhysRevLett.119.180509,PhysRevLett.118.090501,PhysRevX.8.031027,PhysRevResearch.3.033178}. Mathematically, given an HPTP map $\Theta\in\mathscr{C}_{HP}(\cK,\cK^\prime)$, the \emph{physical implementability} is defined as \cite{Jiang_2021_physical}
\begin{align}
    \nu(\Theta)=\log\min\Big\{\eta_1+\eta_2\Big\vert \Theta=&\eta_1\Phi_1-\eta_2\Phi_2,\nonumber\\
    &~\Phi_1,\Phi_2\in\mathscr{C}(\cK,\cK^\prime),~\eta_1,\eta_2\geq0\Big\}.\label{Eq:Phys_imp_def_nu}
\end{align}
The logarithm is in base 2 throughout the remainder of the paper. In Ref. \cite{Jiang_2021_physical}, several desirable properties of $\nu(\Theta)$ are derived. Some of them that will be useful to us in this work are listed as follows
\begin{enumerate}[label=P \arabic*.]
    \item It can be computed efficiently using semidefinite programming and gives the lower bound on the cost of simulating a given HPTP linear map using physically implementable quantum operations.
    \item It is a faithful quantifier. Mathematically, $\nu(\Theta)=0$ \emph{if and only if} $\Theta\in\mathscr{C}(\cH,\cK)$.\label{Eq:Phys_imp_faith}
    \item If we have $\Theta_1\in\mathscr{C}_{HP}(\cH_1,\cK_1)$ and $\Theta_2\in\mathscr{C}_{HP}(\cH_2,\cK_2)$ then
    \begin{align}
        \nu(\Theta_1\otimes\Theta_2)=\nu(\Theta_1)+\nu(\Theta_2).\label{Eq:Phys_impl_add}
    \end{align}
    \item Given two HPTP maps $\Theta\in\mathscr{C}_{HP}(\cH,\cK)$ and $\Theta^{\prime}\in\mathscr{C}_{HP}(\cK,\cK^{\prime})$ we have
    \begin{align}
        \nu(\Theta\circ\Theta^{\prime})\leq\nu(\Theta)+\nu(\Theta^{\prime}).\label{Eq:Phys_imp_subadd}
    \end{align}
    \item Given a superchannel $\Delta$ and a HPTP map $\Theta\in\mathscr{C}_{HP}(\cH,\cK)$ we have
    \begin{align}
        \nu(\Delta[\Theta])\leq\nu(\Theta).
    \end{align}
\end{enumerate}

For two given quantum channels, $\Lambda_1\in\mathscr{C}(\cH_1,\cK_1)$ and $\Lambda_2\in\mathscr{C}(\cH_1,\cK_2)$ such that, $\Lambda_1\succeq_{asymp}\Lambda_2$ (see Remark \ref{Rm:Mult_HPTP_maps}) we can define the \emph{physical implementability} of $\Lambda_2$ given $\Lambda_1$ as
\begin{align}
    \mathfrak{R}_{\Lambda_1}(\Lambda_2):=\log\min_{\Theta}\Big\{2^{\nu(\Theta)}\Big\vert~\Lambda_2=\Theta\circ\Lambda_1,~\Theta\in\mathscr{C}_{HP}(\cK_1,\cK_2)\Big\}.\label{Eq:Def_Phys_imp}
\end{align}
Similar to $\nu(\Theta)$ in Ref. \cite{Jiang_2021_physical}, it can also be characterised by the following semidefinite programming in terms of the Choi matrices of $\Lambda_1, \Lambda_2$, and $\Theta$
\begin{align}
    2^{\mathfrak{R}_{\Lambda_1}(\Lambda_2)}=&\min x_1+x_2\nonumber\\
    &\text{s.t.}\nonumber\\
    &J_\Theta=J_1-J_2\nonumber\\
    &\tr_{\cK_2}J_1\leq x_1\mathbbm{1}_{\cK_1},~\tr_{\cK_2}J_2\leq x_2\mathbbm{1}_{\cK_1}\nonumber\\
    &J_1\geq0,~J_2\geq0\nonumber\\
    &J_{\Lambda_2}=\tr_{\cK_1}[(J_{\Theta}\otimes\mathbbm{1}_{\cH_1})(\mathbbm{1}_{\cK_2}\otimes J_{\Lambda_1}^{T_{\cK_1}})].\label{Eq:Phys_imp_SDP}
\end{align}
 Next, we will list and prove certain important properties $\mathfrak{R}_{\Lambda_1}(\Lambda_2)$

     \begin{proposition}
        \textbf{Faithfulness:} $\mathfrak{R}_{\Lambda_1}(\Lambda_2)=0$ if and only if $\Lambda_1\succeq_{postproc}\Lambda_2$
    \end{proposition}
    \begin{proof}
        Let us start with the \emph{if} part. Assume $\mathfrak{R}_{\Lambda_1}(\Lambda_2)=0$. Then from Eq. \eqref{Eq:Def_Phys_imp} $\exists$ $\Theta^*$ such that $2^{\nu(\Theta^*)}=1$. This implies that $\nu(\Theta^*)=0$. From the Property \hyperref[Eq:Phys_imp_faith]{P2} we can conclude that $\Theta^*$ is a CPTP map and $\Lambda_1\succeq_{postproc}\Lambda_2$.

        Now, for the \emph{only if} part, let us consider $\Lambda_1\succeq_{postproc}\Lambda_2$. So there exists a CPTP map $\cN$ such that $\Lambda_2=\cN\circ\Lambda_1$. From Eq. \eqref{Eq:Phys_imp_def_nu}, we get $\nu(\cN)=0$ which implies $\mathfrak{R}_{\Lambda_1}(\Lambda_2)=0$ by definition. It is clear that $\mathfrak{R}_{\Lambda_1}(\Lambda_2)\geq0$ when $\Lambda_1\succeq_{asymp}\Lambda_2$, as $\nu(\Theta)\geq0$ if $\Theta$ is an HPTP linear map.
    \end{proof}

 \begin{proposition}
     \textbf{Subadditivity:} Given quantum channels $\Lambda_1\in\mathscr{C}(\cH,\cK_1),~ \Lambda_1^{\prime}\in\mathscr{C}(\cH^{\prime},\cK_1^{\prime}),~\Lambda_2\in\mathscr{C}(\cH,\cK_2)$ , and $ \Lambda_2^{\prime}\in\mathscr{C}(\cH^{\prime},\cK_2^{\prime})$, if $\Lambda_1\succeq_{asymp}\Lambda_2$ and $\Lambda_1^{\prime}\succeq_{asymp}\Lambda_2^{\prime}$ holds, then
     \begin{align}
         \mathfrak{R}_{\Lambda_1\otimes\Lambda_1^{\prime}}(\Lambda_2\otimes\Lambda_2^{\prime})\leq\mathfrak{R}_{\Lambda_1}(\Lambda_2)+\mathfrak{R}_{\Lambda_1^{\prime}}(\Lambda_2^{\prime}).
     \end{align}
 \end{proposition}
 \begin{proof}
     As $\Lambda_1\succeq_{asymp}\Lambda_2$ and $\Lambda_1^{\prime}\succeq_{asymp}\Lambda_2^{\prime}$ we have
     \begin{align}
         \Lambda_2=&\Theta^*\circ\Lambda_1,\\
         \Lambda_2^{\prime}=&\Theta^{\prime*}\circ\Lambda_1^{\prime}.
     \end{align}
   Here $\Theta^*\in\mathscr{C}_{HP}(\cK_1,\cK_2)$ and $\Theta^{\prime*}\in\mathscr{C}_{HP}(\cK_1^{\prime},\cK_2^{\prime})$ are the optimal HPTP linear maps for the semidefinite programming for $\mathfrak{R}_{\Lambda_1}(\Lambda_2)$ and $\mathfrak{R}_{\Lambda_1^{\prime}}(\Lambda_2^{\prime})$ respectively. We can then write
   \begin{align}
       \Lambda_2\otimes\Lambda_2^{\prime}=(\Theta^*\otimes\Theta^{\prime*})\circ(\Lambda_1\otimes\Lambda_1^{\prime}).
   \end{align}
   It is clear that $\Theta^*\otimes\Theta^{\prime*}$ is a feasible solution of semidefinite programming for $\mathfrak{R}_{\Lambda_1\otimes\Lambda_1^{\prime}}(\Lambda_2\otimes\Lambda_2^{\prime})$. Let $\tilde{\Theta}$ be the optimal HPTP linear map for the semidefinite programming in Eq. \eqref{Eq:Phys_imp_SDP}.  Then using Property \hyperref[Eq:Phys_impl_add]{P3}, we have
   \begin{align}
       \mathfrak{R}_{\Lambda_1\otimes\Lambda_1^{\prime}}(\Lambda_2\otimes\Lambda_2^{\prime})=&\log2^{\nu(\tilde{\Theta})}\nonumber\\
       \leq&\log2^{\nu(\Theta^*\otimes\Theta^{\prime*})}\nonumber\\
       =&\nu(\Theta^*\otimes\Theta^{\prime*})\nonumber\\
       \leq&\nu(\Theta^*)+\nu(\Theta^{\prime*})\nonumber\\
       =&\mathfrak{R}_{\Lambda_1}(\Lambda_2)+\mathfrak{R}_{\Lambda_1^{\prime}}(\Lambda_2^{\prime}).
   \end{align}
 \end{proof}
 \begin{theorem}
        \textbf{Monotonicity:} Given two CPTP maps $\Lambda_1\in\mathscr{C}(\cH,\cK_1)$ and $\Lambda_2\in\mathscr{C}(\cH,\cK_2)$ with $\Lambda_1\succeq_{asymp}\Lambda_2$, if there exist $\Lambda_1^{\prime}\in\mathscr{C}(\cH,\cK_1^{\prime})$ and $\Lambda_2^{\prime}\in\mathscr{C}(\cH,\cK_2^{\prime})$ such that $\Lambda_1^{\prime}\succeq_{postproc}\Lambda_1$ and $\Lambda_2\succeq_{postproc}\Lambda_2^{\prime}$ holds then
        \begin{align}
            \mathfrak{R}_{\Lambda_1^{\prime}}(\Lambda_2^{\prime})\leq\mathfrak{R}_{\Lambda_1}(\Lambda_2).
        \end{align}
        The equality holds when $\Lambda_1\simeq_{postproc}\Lambda_1^{\prime}$ and $\Lambda_2\simeq_{postproc}\Lambda_2^{\prime}$.
    \end{theorem}
    \begin{proof}
        Given $\Lambda_1^{\prime}\succeq_{postproc}\Lambda_1$ and $\Lambda_2\succeq_{postproc}\Lambda_2^{\prime}$ we can write
        \begin{align}
            &\Lambda_1=\cN_1\circ\Lambda_1^{\prime},\\
            &\Lambda_2^{\prime}=\cN_2\circ\Lambda_2,\label{Eq:Lemma_post_proc}
        \end{align}
        where $\cN_1\in\mathscr{C}(\cK_1^{\prime},\cK_1)$ and $\cN_2\in\mathscr{C}(\cK_2,\cK_2^{\prime})$. As $\Lambda_1\succeq_{asymp}\Lambda_2$, consider $\Theta^*\in\mathscr{C}(\cK_1,\cK_2)$ to be the optimal HPTP linear map for the semidefinite programming for $\mathfrak{R}_{\Lambda_1}(\Lambda_2)$ in Eq. \eqref{Eq:Def_Phys_imp}. So we have
        \begin{align}
            \Lambda_2=\Theta^*\circ\Lambda_1.
        \end{align}
        Then starting from Eq. \eqref{Eq:Lemma_post_proc} we can write
        \begin{align}
            \Lambda_2^{\prime}=&\cN_2\circ\Lambda_2,\nonumber\\
            =&\cN_2\circ\Theta^*\circ\Lambda_1,\nonumber\\
            =&\cN_2\circ\Theta^*\circ\cN_1\circ\Lambda_1^{\prime},\nonumber\\
            =&\Theta^{\prime}\circ\Lambda_1^{\prime},
        \end{align}
        where $\Theta^{\prime}:=(\cN_2\circ\Theta^*\circ\cN_1)\in\mathscr{C}_{HP}(\cK_1^{\prime},\cK_2^{\prime})$. Using the Properties \hyperref[Eq:Phys_imp_faith]{P2} and \hyperref[Eq:Phys_imp_subadd]{P4} of $\nu$ we get
        \begin{align}
            \nu(\Theta^{\prime})=\nu(\cN_2\circ\Theta^*\circ\cN_1)\leq&\nu(\cN_2)+\nu(\Theta^*)+\nu(\cN_1),\nonumber\\
            \leq&\nu(\Theta^*)=\mathfrak{R}_{\Lambda_1}(\Lambda_2).\label{Eq:Lemma3_Phys_imp_ineq}
        \end{align}
        
        Here, the last equality follows from the fact that $\Theta^*$ is the optimal HPTP linear map for the semidefinite programming for $\mathfrak{R}_{\Lambda_1}(\Lambda_2)$. As $\Theta^{\prime}$ is a feasible solution of the semidefinite programming for $\mathfrak{R}_{\Lambda_1^{\prime}}(\Lambda_2^{\prime})$, using Eq. \eqref{Eq:Lemma3_Phys_imp_ineq}, we have
        \begin{align}
            \mathfrak{R}_{\Lambda_1^{\prime}}(\Lambda_2^{\prime})\leq&\nu(\Theta^{\prime}),\nonumber\\
            \leq&\mathfrak{R}_{\Lambda_1}(\Lambda_2).\label{Eq:Phys_Imp_mono_order}
        \end{align}
        It is clear that if $\Lambda_1\simeq_{postproc}\Lambda_1^{\prime}$ and $\Lambda_2\simeq_{postproc}\Lambda_2^{\prime}$ then following the same procedure, we will also get
        \begin{align}
            \mathfrak{R}_{\Lambda_1}(\Lambda_2)\leq \mathfrak{R}_{\Lambda_1^{\prime}}(\Lambda_2^{\prime}).\label{Eq:Phys_Imp_mono_order_rev}
        \end{align}
        Thus, from Eqs. \eqref{Eq:Phys_Imp_mono_order} and \eqref{Eq:Phys_Imp_mono_order_rev}, we see that the equality holds.
\end{proof}
When $\Lambda_1\succeq_{asymp}\Lambda_2$, the quantity $\mathfrak{R}_{\Lambda_1}(\Lambda_2)$, thus, clearly characterizes how far the physical implementability of the quantum channel $\Lambda_2$ given that the quantum channel $\Lambda_1$ has been already implemented deviates from the case when $\Lambda_2$ is just a post processing of $\Lambda_1$.

        In the special case, when $\Lambda_1$ is the post processing of itself \textit{i.e.} $\Lambda_1=\Lambda_1^{\prime}$ the inequality in Eq. \eqref{Eq:Phys_Imp_mono_order} reduces to
        \begin{align}
            \mathfrak{R}_{\Lambda_1}(\Lambda_2^{\prime})
            \leq\mathfrak{R}_{\Lambda_1}(\Lambda_2). 
        \end{align}
        This inequality tells us that whenever $\Lambda_2$ becomes noisier, the physical implementation of $\Lambda_2$, given that the quantum channel $\Lambda_1$ has already been implemented, is easier.

        Similarly, when $\Lambda_2$ is the post-processing of itself, the inequality becomes
        \begin{align}
            \mathfrak{R}_{\Lambda_1^{\prime}}(\Lambda_2)
            \leq\mathfrak{R}_{\Lambda_1}(\Lambda_2). 
        \end{align}
        From this inequality, we conclude that if $\Lambda_1$ is the post-processing of $\Lambda_1^{\prime}$, the physical implementability of the given channel $\Lambda_2$ is easier in the scenario where $\Lambda_1^{\prime}$ has already been implemented rather than in the case where $\Lambda_1$ has already been implemented.

\section{Conclusion}\label{Sec:Conc}
In this work, we have provided a physically meaningful way to compare two quantum channels using Hermitian-preserving trace-preserving linear maps. Whenever given an arbitrary unknown input state, it is possible to determine the output state of one quantum channel $\Lambda_2$ from the output statistics of another quantum channel $\Lambda_1$ using some POVM, we say that the quantum channel $\Lambda_1$ is at least as powerful as the quantum channel $\Lambda_2$ in the asymptotic sense. We have shown that it is equivalent to saying that $\Lambda_2$ can be written as a concatenation of $\Lambda_1$ with a Hermitian-preserving trace-preserving linear map, and this map also does not need to be positive. Thus, for a given pair of quantum channels, we have shown an \emph{operational relevance} to Hermitian-preserving, trace-preserving (HPTP) linear maps. Specifically, all the information required to reconstruct the output of one channel from the other is asymptotically retrievable  when the former channel can be obtained from the latter by the action of an HPTP linear map on its output. We further show that it is possible that $\Lambda_2$ can't be post-processed from $\Lambda_1$ even though the output statistics of the former can be obtained from the output statistics of the latter. Thus, there exists a hierarchy between the preorders induced by these two relations. In between these two, there exists the preorder introduced by the relation between $\Lambda_1$ and $\Lambda_2$, where $\Lambda_2$ can be obtained by concatenating $\Lambda_1$ with a positive linear map. We then characterize the physical implementiblity of $\Lambda_2$ , given $\Lambda_1$ has already been implemented, whenever $\Lambda_2$ can be written as the concatenation of $\Lambda_1$ with a Hermitian-preserving trace-preserving map. We further show the implications of our results on the incompatibility of quantum devices.

Our work opens up several avenues for future research. 
We already know positive maps are important for constructing witnesses for entanglement detection. It would be interesting and worthwhile to see whether a similar operationally meaningful preorder can also be constructed for positive linear maps in the context of the concatenation of quantum channels. Another direction to venture is how some information-theoretic and thermodynamic properties of two channels are related to each other, given that one of them is at least as powerful as the other in an asymptotic sense. It will also be interesting to explore how our study is helpful in quantifying the performance of some information-theoretic or thermodynamic tasks.

\section{Acknowledgments}
AM acknowledges STARS (Grant No.
STARS/STARS-2/2023-0809), Government of India for support. The authors thank Prof Guruprasad Kar, Prof. Siddhartha Das, Dr. Martin Pl{\'a}vala, Prof. Teiko Heinosaari, and Prof. Francesco Buscemi for their useful comments.

\bibliography{reference}

\end{document}